\documentclass[a4paper]{article}

\usepackage{a4wide}
\usepackage{amsmath,amssymb,amsthm}
\usepackage{graphicx}

\usepackage[ruled]{algorithm}
\usepackage{algorithmicx}
\usepackage{paralist}
\usepackage[noend]{algpseudocode}
\usepackage{color}
\usepackage{ifthen}

\title{Online Algorithms for Multi-Level Aggregation%
\thanks{Research partially supported by NSF grants CCF-1536026, CCF-1217314 and OISE-1157129,
	Polish NCN grants DEC-2013/09/B/ST6/01538, 2015/18/E/ST6/00456,
	project 14-10003S of GA \v{C}R and GAUK project 548214.
}}

\author{Marcin Bienkowski%
\thanks{Institute of Computer Science, University of Wroc{\l}aw, Poland}
\and
Martin B\"{o}hm%
\thanks{Computer Science Institute, Charles University, Czech Republic}
\and
Jaroslaw Byrka\footnotemark[2]
\and
Marek Chrobak%
\thanks{Department of Computer Science, University of California at
  Riverside, USA}
\and
Christoph D\"{u}rr%
\thanks{Sorbonne Universit\'{e}s, UPMC Univ Paris 06, CNRS, LIP6, Paris, France}
\and
Luk\'{a}\v{s} Folwarczn\'{y}\footnotemark[3]
\and
{\L}ukasz Je\.{z}\footnotemark[2]
\and
Ji\v{r}\'{\i} Sgall\footnotemark[3]
\and
Nguyen Kim Thang%
\thanks{IBISC, Universit\'{e} d'Evry Val d'Essonne, France}
\and
Pavel Vesel\'{y}\footnotemark[3]}

\newtheorem{theorem}{Theorem}[section]

\newtheorem{lemma}[theorem]{Lemma}
\newtheorem{corollary}[theorem]{Corollary}

\newcommand{\etal}{et~al.}

\newcommand{\barO}{\overline{O}}

\newcommand{\barX}{\overline{X}}
\newcommand{\barY}{\overline{Y}}

\newcommand{\calA}{{\cal A}}
\newcommand{\calB}{{\cal B}}

\newcommand{\calJ}{{\cal J}}

\newcommand{\calP}{{\cal P}}
\newcommand{\calR}{{\cal R}}

\newcommand{\calT}{{\cal T}}

\newcommand{\bbI}{\mathbb{I}}

\newcommand{\eps}{{\varepsilon}}

\newcommand{\e}{\mathrm{e}}

\newcommand{\braced}[1]{{ \left\{ #1 \right\} }}
\newcommand{\angled}[1]{{ \left\langle #1 \right\rangle }}

\newcommand{\cost}{\mbox{\rm\textsf{cost}}}

\newcommand{\wcost}{\mbox{\rm\textsf{wcost}}}
\newcommand{\scost}{\mbox{\rm\textsf{scost}}}
\newcommand{\opt}{\mbox{\rm\textsf{opt}}}
\newcommand{\alg}{\mbox{\rm\textsf{alg}}}

\newcommand{\length}{\ell}
\newcommand{\depth}{\textit{depth}}
\newcommand{\parent}{\textit{parent}}
\newcommand{\surplus}{{\delta}}

\newcommand{\reals}{{\mathbb R}}

\newcommand{\half}{{\mbox{$\frac{1}{2}$}}}

\newcommand{\algDoubling}{\mbox{\sc OnlDoubling}}
\newcommand{\algCovSubT}{\mbox{\sc CovSubT}}
\newcommand{\algLBL}{\mbox{\sc OffLByL}}
\newcommand{\OnAlgTreesDeadlines}{{\sc OnlTreeD}}
\newcommand{\OnAlgTreesGeneral}{{\sc OnlTree}}
\newcommand{\DLINE}{\mbox{\textsc{OnlLine}}}

\newcommand{\NP}{{\mathbb{NP}}}
\newcommand{\APX}{{\mathbb{APX}}}
\newcommand{\JRPD}{\mbox{\rm\textsf{JRP-D}}}
\newcommand{\JRP}{\mbox{\rm\textsf{JRP}}}
\newcommand{\TCPAP}{\mbox{\rm\textsf{TCP-AP}}}
\newcommand{\MLAP}{\mbox{\rm\textsf{MLAP}}}
\newcommand{\MLAPL}{\mbox\rm{\textsf{MLAP-L}}}
\newcommand{\MLAPD}{\mbox{\rm\textsf{MLAP-D}}}
\newcommand{\SPMLAP}{\mbox{\rm\textsf{1P-MLAP}}}
\newcommand{\SPMLAPD}{\mbox{\rm\textsf{1P-MLAP-D}}}
\newcommand{\SPMLAPL}{\mbox{\rm\textsf{1P-MLAP-L}}}

\newcommand{\schedS}{\mbox{\rm\textsf{S}}}
\newcommand{\optschedS}{\mbox{\rm\textsf{S}}^\ast}
\newcommand{\pseudoschedS}{\overline{\mbox{\rm\textsf{S}}}}

\newcommand{\expiration}{\theta}

\newcommand{\shift}{{\textit{shift}}}

\newcommand{\trignode}{\sigma}
\newcommand{\advtrans}{\mbox{\rm\textsf{nos}}}
\newcommand{\urgentnodes}{\mbox{\rm\textsf{Urgent}}}
\newcommand{\treematurity}{\mu}
\newcommand{\vertmaturity}{M}

\newcommand{\prv}[2]{\mbox{\rm Prev}^#1(#2)}
\newcommand{\nxt}[3]{\mbox{\rm Next}^#1(#2,#3)}
\newcommand{\matvt}{{t^\ast}}

\newcommand{\ALG}{\textsc{Alg}}

\newcommand{\onefourth}{{\textstyle\frac{1}{4}}}

\newcommand{\onefifth}{{\textstyle\frac{1}{5}}}

\begin{document}

\maketitle

\begin{abstract}
In the \emph{Multi-Level Aggregation Problem} ({\MLAP}), requests arrive at
the nodes of an edge-weighted tree $\calT$, and have to be served 
eventually.  A \emph{service} is defined as a subtree $X$ of $\calT$ that 
contains its root. This subtree~$X$ serves all
requests that are pending in the nodes of $X$, and the cost of this service is
equal to the total weight of $X$. Each request also incurs waiting cost
between its arrival and service times. The objective is to minimize the total
waiting cost of all requests plus the total cost of all service subtrees.
{\MLAP} is a generalization of some well-studied optimization problems; for
example, for trees of depth $1$, {\MLAP} is equivalent to the TCP
Acknowledgment Problem, while for trees of depth $2$, it is equivalent to the
Joint Replenishment Problem. Aggregation problems for trees of arbitrary depth
arise in multicasting, sensor networks, communication in organization
hierarchies, and in supply-chain management. The instances of {\MLAP}
associated with these applications are naturally online, in the sense that
aggregation decisions need to be made without information about future
requests.

Constant-competitive online algorithms are known for {\MLAP} with one or two
levels. However, it has been open whether there exist constant competitive
online algorithms for trees of depth more than $2$.  Addressing this open
problem, we give the first constant competitive online algorithm for trees
of arbitrary (fixed) depth. The competitive ratio is $O(D^42^D)$,
where $D$ is the depth of $\calT$. The algorithm works for arbitrary waiting
cost functions, including the variant with deadlines.
We include several additional results in the paper. We show that a standard
lower-bound technique for {\MLAP}, based on so-called \emph{Single-Phase}
instances, cannot give super-constant lower bounds (as a function of the tree
depth). This result is established by giving an online algorithm with optimal
competitive ratio $4$ for such instances on arbitrary trees. 
We prove that, in the offline case, these instances can be solved to
optimality in polynomial time.
We also study the {\MLAP} variant when the tree is a path, for which we give
a lower bound of $4$ on the competitive ratio, improving the lower bound known 
for general {\MLAP}.  We complement this with a matching upper bound for the 
deadline setting. 
In addition, for arbitrary trees,
we give a simple 2-approximation algorithm for offline {\MLAP} 
with deadlines.
\end{abstract}



\section{Introduction}
\label{sec: introduction}

Certain optimization problems can be formulated as aggregation problems. They
typically arise when expensive resources can be shared by multiple agents, who
incur additional expenses for accessing a resource. 
For example, costs may be associated with
waiting until the resource is accessible, or, if the resource is not in the
desired state, a costly setup or retooling may be required.

\paragraph{1-level aggregation.}

A simple example of an aggregation problem is the \emph{TCP Acknowledgment
Problem ($\TCPAP$)}, where control messages (``agents'') waiting for transmission across a
network link can be aggregated and transmitted in a~single packet (``resource''). 
Such aggregation can
reduce network traffic, but it also results in undesirable delays. A reasonable compromise
is to balance the two costs, namely the number of transmitted
packets and the total delay, by minimizing their weighted sum~\cite{tcp-ack-det-journal}.
Interestingly, $\TCPAP$ is equivalent to the classical Lot Sizing Problem studied in the operations research literature
since the 1950s. (See, for example, \cite{wagner_whitin_58}.) In the offline variant of $\TCPAP$, that is
when all arrival times of control messages are known beforehand, an optimal schedule for
aggregated packets can be computed with dynamic programming in time $O(n\log n)$~\cite{aggarwal_park_93}.
In practice, however, packet aggregation decisions must be done on the fly,
without any information about future message releases. This scenario is captured by the
online variant of $\TCPAP$ that has also been well studied; it is known that the optimal competitive
ratio is $2$ in the deterministic case \cite{tcp-ack-det-journal}
and $\mathrm{e}/(\mathrm{e}-1) \approx 1.582$ in the randomized
case~\cite{tcp-ack,online-primal-dual-book,tcp-ack-lower-bound}.
Online variants of $\TCPAP$ that use different assumptions or objective functions 
were also examined in the literature~\cite{tcp-ack-logp,tcp-ack-long-delays}.


\paragraph{2-level aggregation.}

Another optimization problem involving aggregation is the \emph{Joint Replenishment Problem
($\JRP$)}, well-studied in operations research.
$\JRP$ models tradeoffs that arise in supply-chain management.
One such scenario involves optimizing shipments of goods from a
supplier to  retailers, through a shared warehouse, in response to their
demands.  In $\JRP$, aggregation takes place at two levels: items addressed to
different retailers can be shipped together to the warehouse, at a fixed cost,
and then multiple items destined to the same retailer can be shipped  from the warehouse to
this retailer together, also at a fixed cost, which can be  different for different
retailers. Pending demands accrue waiting cost until they are satisfied by a shipment.
The objective is to minimize the sum of all shipment costs and all waiting costs.

$\JRP$ is known to be $\NP$-hard~\cite{jrp-arkin}, and
even $\APX$-hard~\cite{jrp-deadlines-nonner,bienkowski_jrpd_2013}.
The currently best approximation,
due to Bienkowski~{\etal}~\cite{jrp-soda-2014}, achieves a factor of $1.791$,
improving on earlier work by
Levi~{\etal}~\cite{jrp-owmr-levi-soda,jrp-owmr-levi-journal,jrp-owmr-levi-approx}.
In the deadline variant of $\JRP$, denoted $\JRPD$,
there is no cost for waiting, but each demand needs to be satisfied before its deadline.
As shown in~\cite{bienkowski_jrpd_2013}, $\JRPD$ can be approximated with ratio $1.574$.

For the online variant of $\JRP$,
Buchbinder~{\etal}~\cite{jrp-online-buchbinder} gave a $3$-competitive algorithm using
a primal-dual scheme (improving an earlier bound of $5$ in~\cite{aggregation-bkv})
and proved a lower bound of $2.64$, that was subsequently
improved to $2.754$~\cite{jrp-soda-2014}.
The optimal competitive ratio for $\JRPD$ is $2$~\cite{jrp-soda-2014}.


\paragraph{Multiple-level aggregation.}  

$\TCPAP$ and $\JRP$ can be
thought of as aggregation problems on edge-weighted trees of depth $1$ and $2$,
respectively. In $\TCPAP$, this tree is just a single edge between the sender 
and the recipient. In $\JRP$, this tree consists of the root (supplier),
with one child (warehouse), and any number of grandchildren (retailers).
A shipment can be represented by a subtree of this tree and edge weights represent
shipping costs. These trees capture the general
problem on trees of depth $1$ and $2$, as the children of the root can
be considered separately (see Section~\ref{sec: preliminaries}).

This naturally
extends to trees of any depth $D$, where aggregation is allowed at each level.
Multi-level message aggregation has been, in fact, studied in 
communication networks in several contexts. In multicasting, protocols for aggregating
control messages (see \cite{bortnikov_cohen_infocom_98,badrinath_sudame_gathercast_00}, for example)
can be used to reduce the so-called \emph{ack-implosion}, the
proliferation of control messages routed to the source.
A similar problem arises in energy-efficient data aggregation and fusion in
sensor networks~\cite{hu_cao_may_sensors_05,yuan_et_al_data_fusion_03}.
Outside of networking, tradeoffs between the cost of communication and delay arise in
message aggregation in organizational hierarchies~\cite{Papadimitriou_96}.
In supply-chain management, multi-level variants of lot sizing have been
studied~\cite{wallance_et_al_multistage_assembly_73,kims_multilevel_lot_sizing_97}.
The need to consider more tree-like (in a broad sense) supply hierarchies
has also been advocated in~\cite{Lambert_Cooper_issues_chain_management_00}.

These applications have inspired research on offline and online
approximation algorithms for multi-level aggregation problems.
Becchetti~{\etal}~\cite{packet-aggregation-becchetti} gave a $2$-approximation
algorithm for the deadline case. (See also~\cite{aggregation-bkv}.)
Pedrosa~\cite{lehilton-note} showed, adapting an algorithm of
Levi~{\etal}~\cite{jrp-levi-2-approx} for the multi-stage assembly problem,
that there is a $(2+\eps)$-approximation algorithm for general waiting cost functions, 
where $\eps$ can be made arbitrarily small.

In the online case,
Khanna~{\etal}~\cite{khanna-message-aggregation} gave a rent-or-buy solution
(that serves a group of requests once their waiting cost reaches the cost
of their service) and showed that their algorithm is $O(\log \alpha)$-competitive,
where $\alpha$ is defined as the sum of all edge weights. 
However, they assumed that
each request has to wait at least one time unit. This assumption is crucial for their
proof, as demonstrated by Brito~{\etal}~\cite{aggregation-bkv}, who showed that
the competitive ratio of a rent-or-buy strategy is $\Omega(D)$, even
for paths with $D$ edges. 
The same assumption of a minimal cost for a request and a ratio dependent
on the edge-weights is also essential in the work of
Vaya~\cite{Vaya_delay_deliver_12}, who studies a variant of the
problem with bounded bandwidth (the number of packets that can be
served by a single edge in a single service).

The existence of a primal-dual $(2+\eps)$-approximation algorithm \cite{lehilton-note, jrp-levi-2-approx}
for the offline problem suggests the possibility of constructing an online algorithm along the lines of \cite{online-primal-dual-book}. Nevertheless, despite substantial effort of many researchers,
the online multi-level setting remains wide open.
This is perhaps partly due to impossibility of direct emulation of
the cleanup phase in primal-dual offline algorithms in the online setting, 
as this cleanup is performed in the ``reverse time'' order.

The case when the tree is just a path has also been studied.
An offline polynomial-time algorithm that computes an optimal schedule was given in~\cite{aggregation_wads_2013}.
For the online variant,
Brito {\etal}~\cite{aggregation-bkv} gave an $8$-competitive algorithm.
This result was improved by Bienkowski {\etal}~\cite{aggregation_wads_2013} who showed that
the competitive ratio of this problem is between $2+\phi\approx 3.618$ and $5$.


\subsection{Our Contributions}

We study online competitive algorithms for multi-level aggregation. Minor technical
differences notwithstanding, our model is equivalent to those studied
in~\cite{aggregation-bkv,khanna-message-aggregation}, also extending the
deadline variant in~\cite{packet-aggregation-becchetti} and the assembly
problem in~\cite{jrp-levi-2-approx}. We have decided to choose a more generic
terminology to emphasize general applicability of our model and techniques.

Formally, our model consists of a tree $\calT$ with positive weights assigned to edges,
and a set $\calR$ of requests that arrive in the nodes of $\calT$ over time. These
requests are served by subtrees rooted at the root of $\calT$. Such a
subtree $X$ serves all requests pending at the nodes of $X$ at cost equal to
the total weight of $X$. Each request incurs a waiting cost, defined by 
a~non-negative and non-decreasing function of time, which may be different for
each request. The objective is to minimize the sum of the total
service and waiting costs. We call this the \emph{Multi-Level Aggregation
Problem} ({\MLAP}).

In most earlier papers on aggregation problems, the waiting cost function is
linear, that is, it is
assumed to be simply the delay between the times when a request arrives and
when it is served. We denote this version by  {\MLAPL}.
However, most of the algorithms for this model extend naturally to arbitrary cost
functions.
Another variant is  {\MLAPD}, where each request is given a certain deadline,
has to be served before or at its deadline, and there is no penalty associated
with waiting. This can be modeled by the waiting cost function that is $0$ up to
the deadline and $+\infty$ afterwards.

In this paper, we mostly focus on the online version of {\MLAP}, where an algorithm needs to
produce a schedule in response to requests that arrive over time.
When a request appears, its waiting cost function is
also revealed. At each time $t$, the online algorithm needs to decide whether to
generate a service tree at this time, and if so, which nodes should be included in this tree.


\begin{table}
\begin{center}
\begin{tabular}{r|l|l|l|l|}
	& \multicolumn{2}{c|}{\MLAP\ and \MLAPL} & \multicolumn{2}{c|}{\MLAPD} \\
\hline
 & upper & lower & upper & lower \\
\hline
depth 1 & $2^*$~\cite{tcp-ack-det-journal} & 2~\cite{tcp-ack-det-journal}
& 1 & 1 \\
rand.~alg.~for depth 1 & $1.582^*$~\cite{tcp-ack} &
1.582~\cite{tcp-ack-lower-bound} &
1 & 1 \\
depth 2 & 3~\cite{jrp-online-buchbinder} & 2.754~\cite{jrp-soda-2014} & 2~\cite{jrp-soda-2014} & 2~\cite{jrp-soda-2014} \\
fixed depth $D \geq 2$ & $\mathbf{O(D^42^D)}$ & 2.754 & $\mathbf{D^22^D}$ & 2 \\
paths of arbitrary depth & $5^*$~\cite{aggregation_wads_2013} & 3.618~\cite{aggregation_wads_2013}, {\bf 4} & {\bf 4} & {\bf 4} \\
\hline
\end{tabular}
\end{center}
\caption{Previous and current bounds on the competitive ratios for {\MLAP} for trees of various depths.
Ratios written in bold are shown in this paper. Unreferenced results are either immediate consequences of
other entries in the table or trivial observations.
Asterisked ratios represent results for {\MLAP} with arbitrary waiting cost functions, which, though
not explicitly stated in the respective papers,  are straightforward extensions of the
corresponding results for {\MLAPL}. Some values in the table are approximations:
$1.582$ represents $\mathrm{e}/(\mathrm{e}-1)$ and $3.618$ represents $2+\phi$, where $\phi$ is the golden ratio.
}
\label{tab:results}
\end{table}

The main result of our paper is an $O(D^42^D)$-competitive algorithm
for {\MLAP} for trees of depth $D$, presented in 
Section~\ref{sec: competitive algorithm for mlap}.  A simpler $D^22^D$-competitive
algorithm for {\MLAPD} is presented in Section~\ref{sec: competitive algorithm for mlap-d}. 
No competitive algorithms have been known so far
for online {\MLAP} for arbitrary depth trees, even for the special case of
{\MLAPD} on trees of depth $3$.

For both results we use a reduction, described in Section~\ref{sec: reduction to L-decreasing trees},
of the general problem to the special case of trees with fast decreasing weights described. 
For such trees we then provide an explicit competitive algorithm. 
While our algorithm is compact and elegant, it is not a straightforward
extension of the 2-level algorithm.  (In fact, we have been able to show that na\"{\i}ve 
extensions of the latter algorithm are not competitive.)
It is based on carefully constructing a sufficiently large service tree
whenever it appears that an urgent request must be served. The specific structure of the service 
tree is then heavily exploited in an amortization argument that constructs a
mapping from the algorithm's cost to the cost of the optimal schedule.
We believe that these three new techniques: the reduction to trees with fast decreasing weights,
the construction of the service trees, and our charging scheme, 
will be useful in further studies of online aggregation problems.

In Section~\ref{sec: one-phase MLAP} we study a version of
$\MLAP$, that we refer to as \emph{Single-Phase} {\MLAP} (or
$\SPMLAP$), in which all requests arrive at the beginning, but they
also have a common \emph{expiration time} that we denote by
$\expiration$.  Any request not served by time~$\expiration$ pays
waiting cost at time~$\expiration$ and does not need to be served
anymore. In spite of the expiration-date feature, it can be shown that
$\SPMLAP$ can be represented as a special case of $\MLAP$.  
$\SPMLAP$ is a crucial tool in all the lower bound proofs in the literature
for competitive ratios of {\MLAP}, including those in~\cite{jrp-online-buchbinder,aggregation_wads_2013}, 
as well as in our lower bounds in Section~\ref{sec: mlap on paths}. 
It also has a natural interpretation in the context of $\JRP$ ($2$-level $\MLAP$), if we
allow all orders to be canceled, say, due to changed market
circumstances.  In the online variant of $\SPMLAP$ all requests are
known at the beginning, but the expiration time $\expiration$ is
unknown.  For this version we give an online algorithm with
competitive ratio $4$, matching the lower bound.
Since $\SPMLAP$ can be expressed as a special case of $\MLAP$, our result
implies that the techniques
from~\cite{jrp-online-buchbinder,aggregation_wads_2013} cannot be used
to prove a lower bound larger than $4$ on the competitive ratio for $\MLAP$,
and any study of the dependence of the competitive
ratio on the depth $D$ will require new insights and techniques.

In Section~\ref{sec: mlap on paths} we consider $\MLAP$ on paths.
For this case, we give a $4$-competitive algorithm for {\MLAPD} and we provide
a matching lower bound.  We show that the same lower bound of $4$ applies to 
$\MLAPL$ as well, improving the previous lower bound of $3.618$
from~\cite{aggregation_wads_2013}.

In addition, we provide two results on offline algorithms (for arbitrary trees).  
In Section~\ref{sec: mlap with deadlines} we provide a 2-approximation
algorithm for $\MLAPD$, significantly simpler than the LP-rounding
algorithm in~\cite{packet-aggregation-becchetti} with the same ratio.
In Section~\ref{sec: 1p-mlap with deadlines}, we give a polynomial
time algorithm that computes optimal solutions for $\SPMLAP$.

Finally, in Section~\ref{sec: general waiting costs}, we discuss
several technical issues concerning the use of general functions as
waiting costs in $\MLAP$. In particular, when presenting our algorithms for $\MLAP$
we assume that all waiting cost functions are continuous (which cannot directly capture some
interesting variants of $\MLAP$).
This is done, however, only for technical convenience; as explained in Section~\ref{sec: general waiting costs},
these algorithms can be extended to left-continuous
functions, which allows to model {\MLAPD} as a special case
of {\MLAP}. We also consider two alternative models for {\MLAP}: the
discrete-time model and the model where not all requests need to be
served, showing that our algorithms can be extended to these models as well.

An extended abstract of this work appeared in the proceedings of  
24th Annual European Symposium on Algorithms (ESA'16)~\cite{Bienkowski_etal_multilevel_esa_2016}.


\section{Preliminaries}
\label{sec: preliminaries}

\paragraph{Weighted trees.}

Let $\calT$ be a tree with root $r$.
For any set of nodes $Z\subseteq\calT$ and a node $x$, $Z_x$ denotes
the set of all descendants of $x$ in $Z$; in particular, $\calT_x$ is 
the \emph{induced subtree} of $\calT$ rooted at $x$.
The parent of a node $x$ is denoted $\parent(x)$.
The \emph{depth of $x$}, denoted $\depth(x)$, is
the number of edges on the simple path from $r$ to $x$. In
particular, $r$ is at depth $0$. The depth $D$ of $\calT$ is the
maximum depth of a node of $\calT$.

We will deal with weighted trees in this paper. For $x\neq r$, by
$\length_x$ or $\length(x)$ we denote the weight of the edge
connecting node $x$ to its parent. For the sake of convenience,
we will often refer to $\length_x$ as the weight of $x$.
We assume that all these weights
are positive. We extend this notation to $r$ by setting $\length_r=0$.
If $Z$ is any set of nodes of $\calT$, then the weight of $Z$ is
$\length(Z) = \sum_{x\in Z}\length_x$.


\paragraph{Definition of $\MLAP$.}

A \emph{request} $\rho$ is specified by a triple $\rho =
(\trignode_\rho,a_\rho,\omega_\rho)$, where $\trignode_\rho$ is the
node of $\calT$ in which $\rho$ is issued, $a_\rho$ is the non-negative
\emph{arrival time} of $\rho$, and $\omega_\rho$ is the waiting cost
function of $\rho$. We assume that $\omega_\rho(t) = 0$ for $t\le
a_\rho$ and $\omega_\rho(t)$ is non-decreasing for $t\ge a_\rho$. 
$\MLAPL$ is the variant of $\MLAP$ with linear waiting costs; that is,
for each request $\rho$ we have
$\omega_\rho(t) = t-a_\rho$, for $t\ge a_\rho$.
In $\MLAPD$, the variant with deadlines,
we have $\omega_\rho(t) = 0$ for $t\le d_\rho$ and
$\omega_\rho(t) = \infty$ for $t > d_\rho$, where $d_\rho$ is called
the \emph{deadline} of request $\rho$.

In our algorithms for $\MLAP$ with general costs we will be assuming that 
all waiting cost functions are continuous.
This is only for technical convenience and we discuss more general waiting cost
functions in Section~\ref{sec: general waiting costs};
we also show there that {\MLAPD} can be considered a special case of {\MLAP},
and that our algorithms can be extended to the discrete-time model.

A \emph{service} is a pair $(X,t)$, where $X$ is a subtree of $\calT$
rooted at $r$ and $t$ is the time of this service. We will
occasionally refer to $X$ as the service tree (or just service) at
time $t$, or even omit $t$ altogether if it is understood from
context.

An instance $\calJ = \angled{\calT,\calR}$ of the \emph{Multi-Level
Aggregation Problem} ({\MLAP}) consists of a weighted tree $\calT$
with root $r$ and a set $\calR$ of requests arriving at the nodes of
$\calT$.  A \emph{schedule} is a set $\schedS$ of services. For a
request $\rho$, let $(X,t)$ be the service in $\schedS$ with minimal
$t$ such that $\trignode_\rho\in X$ and $t\geq a_\rho$.  We then say
that $(X,t)$ \emph{serves} $\rho$ and the \emph{waiting cost} of
$\rho$ in $\schedS$ is defined as
$\wcost(\rho,\schedS)=\omega_\rho(t)$. Furthermore, the request $\rho$
is called \emph{pending} at all times in the interval $[a_\rho,t]$.
Schedule $\schedS$ is called \emph{feasible} if all requests in
$\calR$ are served by $\schedS$.

The cost of a feasible schedule $\schedS$, denoted $\cost(\schedS)$, is
defined by
\begin{equation*}
	\cost(\schedS) =  \scost(\schedS)+\wcost(\schedS),
\end{equation*}
where $\scost(\schedS)$ is the total service cost and $\wcost(\schedS)$ is the total waiting cost, that is
\begin{equation*}
	\scost(\schedS) = \sum_{(X,t)\in \schedS} \length(X) 
	\quad \textrm{and} \quad
	\wcost(\schedS) = \sum_{\rho\in \calR} \wcost(\rho,\schedS).
\end{equation*}
The objective of $\MLAP$ is to compute a feasible schedule $\schedS$
for $\calJ$ with minimum $\cost(\schedS)$.


\paragraph{Online algorithms.}

We use the standard and natural definition of online algorithms and the competitive ratio. 
We assume the continuous time model. The computation starts at time $0$ and from then on the
time gradually progresses. At any time $t$ new requests can arrive. If the current time is
$t$, the algorithm has complete information about the requests that arrived up until time $t$,
but has no information about any requests whose arrival times are after time $t$.
The instance includes a time horizon $H$ that is not known to the online algorithm, which
is revealed only at time $t = H$. At time $H$, all requests that are still pending must be served.
(In the offline case, $H$ can be assumed to be equal to the maximum request arrival time.)

If $\calA$ is an online algorithm and $R\ge 1$, we say that $\calA$ is \emph{$R$-competitive}%
\footnote{Definitions of competitiveness in the literature often
allow an additive error term, independent of the request sequence. 
For our algorithms, this additive term is not needed. Our lower bound proofs can be
easily modified (essentially, by iterating the adversary strategy) to remain valid
if an additive term is allowed, even if it is a function of $\calT$.}
if $\cost(\schedS)\le R\cdot \opt(\calJ)$ for any instance $\calJ$ of $\MLAP$,
where $\schedS$ is the schedule computed by $\calA$ on $\calJ$ and $\opt(\calJ)$ is the
optimum cost for $\calJ$.


\paragraph{Quasi-root assumption.}

Throughout the paper we will
assume that $r$, the root of $\calT$, has only one child. This
is without loss of generality, because if we have an algorithm (online
or offline) for $\MLAP$ on such trees, we can apply it independently to each
child of $r$ and its subtree. This will give us an algorithm for $\MLAP$ on
arbitrary trees with the same performance. From now on, let us call
the single child of $r$ the \emph{quasi-root} of $\calT$ and denote it by
$q$.  Note that $q$ is included in every (non-trivial) service.


\paragraph{Urgency functions.}

When choosing nodes for inclusion in a service, our online algorithms
give priority to those that are most ``urgent''.  For {\MLAPD},
naturally, urgency of nodes can be measured by their deadlines, where
a deadline of a node $v$ is the earliest deadline of a request pending
in the subtree $\calT_v$, i.e., the induced subtree rooted at $v$.
But for the arbitrary instances of {\MLAP} we need a more general
definition of urgency, which takes into account the rate of increase
of the waiting cost in the future.  To this end, each of our
algorithms will use some \emph{urgency function} $f:\calT\rightarrow
\reals\cup\{+\infty\}$, which also depends on the set of pending
requests and the current time step, and which assigns some time value
to each node. The earlier this value, the more urgent the node is.

Fix some urgency function $f$. Then,
for any set $A$ of nodes in $\calT$ and a real number $\beta$, let
$\urgentnodes(A,\beta,f)$ be the set of nodes obtained by choosing the
nodes from $A$ in order of their increasing urgency value, until either their total weight exceeds $\beta$
or we run out of nodes.
More precisely, we define $\urgentnodes(A,\beta,f)$ as the
smallest set of nodes in $A$ such that 
(i)  for all $u\in \urgentnodes(A,\beta,f)$, and $v\in
A-\urgentnodes(A,\beta,f)$ we have $f(u)\leq f(v)$, 
and 
(ii) either $\length(\urgentnodes(A,\beta,f))\ge\beta$ or
$\urgentnodes(A,\beta,f)=A$.
In case of ties in the values of $f$ there may be multiple choices for
$A$; we choose among them arbitrarily.


\section{Reduction to $L$-Decreasing Trees}
\label{sec: reduction to L-decreasing trees}

One basic intuition that emerges from earlier works on trees of depth
$2$ (see \cite{jrp-online-buchbinder,aggregation-bkv,jrp-soda-2014})
is that the hardest case of the problem is when $\length_q$, the
weight of the quasi-root, is much larger than the weights of
leaves. For arbitrary depth trees, the hard case is when the weights
of nodes quickly decrease with their depth. We show that this is
indeed the case, by defining the notion of $L$-decreasing trees that
captures this intuition and showing that {\MLAP} reduces to the
special case of {\MLAP} for such $L$-decreasing trees, increasing the
competitive ratio by a factor of at most $DL$.  This is a general
result, not limited only to algorithms in our paper.

Formally, for $L\ge 1$, we say
that $\calT$ is \emph{$L$-decreasing} if for each node $u\neq r$ and
each child $v$ of $u$ we have $\length_u \ge L\cdot \length_v$.
(The value of $L$ used in our algorithms will be fixed later.) 
 
Note that the $L$-decreasing condition corresponds to the usual
definition of hierarchically well-separated trees (HSTs); however, for our
purposes we do not need any balancing condition usually also required from HSTs.


\begin{theorem}\label{thm:reduction}
Assume that there exists an $R$-competitive algorithm $\calA$ for
{\MLAP} (resp.~{\MLAPD}) on $L$-decreasing trees (where $R$ can be a
function of $D$, the tree depth).  Then there exists a
$(DLR)$-competitive algorithm $\calB$ for {\MLAP} (resp.~{\MLAPD}) on
arbitrary trees.
\end{theorem}

\begin{proof}
Fix the underlying instance $\calJ = (\calT,\calR)$, where $\calT$ is
a tree and $\calR$ is a sequence of requests in $\calT$.  In our
reduction, we convert $\calT$ to an $L$-decreasing tree $\calT'$ on
the same set of nodes. We then show that any service on $\calT$ is
also a service on $\calT'$ of the same cost and, conversely, that any
service on $\calT'$ can be converted to a slightly more expensive
service on $\calT$.

We start by constructing 
an $L$-decreasing tree $\calT'$ on the same set of nodes. For any node
$u\in \calT-\braced{r}$, the parent of $u$ in $\calT'$ will be the
lowest (closest to $u$) ancestor $w$ of $u$ in $\calT$ such that
$\length_w \geq L\cdot \length_u$; if no such $w$ exists, we take
$w=r$. Note that $\calT'$ may violate the quasi-root assumption, which
does not change the validity of the reduction, as we may use
independent instances of the algorithm for each child of $r$ in
$\calT'$.
Since in $\calT'$ each node $u$ is connected to one of its ancestors
from $\calT$, it follows that $\calT'$ is a tree rooted at $r$ with
depth at most $D$. Obviously, $\calT'$ is $L$-decreasing.

The construction implies that if a set of nodes $X$
is a service subtree of $\calT$, then it is also a service subtree for
$\calT'$.  (However, note that the actual topology of
the trees with node set $X$ in $\calT$ and $\calT'$ may be very different.
For example, if $L=5$ and $\calT$ is a path with costs
(starting from the leaf) $1, 2, 2^2, ..., 2^D$, then in $\calT'$ the
node of weight $2^i$ is connected to the node of weight $2^{i+3}$,
except for the last three nodes that are connected to $r$. Thus the
resulting tree consists of three paths ending at $r$ with roughly
the same number of nodes.)
Therefore, any schedule for $\calJ$ is also a schedule for $\calJ' =
(\calT',\calR)$, which gives us that $\opt(\calJ')\leq \opt(\calJ)$.

The algorithm $\calB$ for $\calT$ is defined as follows: On a request
sequence $\calR$, we simulate $\calA$ for $\calR$ in $\calT'$, and
whenever $\calA$ contains a service $X$, $\calB$ issues the service
$X' \supseteq X$, created from $X$ as follows:
Start with $X'=X$. Then,
for each $u\in X-\braced{r}$, if $w$ is the parent of $u$ in $\calT'$, then
add to $X'$ all inner nodes on the path from $u$ to $w$ in
$\calT$. By the construction of $\calT'$, for each $u$ we add at most
$D-1$ nodes, each of weight less than $L\cdot \length_u$. It follows that
$\length(X')\leq ((D-1)L+1)\length(X)\leq DL\cdot \length(X)$.

In total, the service cost of $\calB$ is at most $DL$ times the
service cost of $\calA$. Any request served by $\calA$ is served by
$\calB$ at the same time or earlier, thus the waiting cost of $\calB$
is at most the waiting cost of $\calA$ (resp.~for {\MLAPD}, $\calB$
produces a valid schedule for $\calJ$). 
Since $\calA$ is $R$-competitive, we obtain 
\[
\cost(\calB,\calJ)\leq DL\cdot\cost(\calA,\calJ')
			\leq DLR\cdot \opt(\calJ')
			\leq DLR \cdot \opt(\calJ),
\]
and thus $\calB$ is $DLR$-competitive.
\end{proof}


\section{A Competitive Algorithm for {\MLAPD}}
\label{sec: competitive algorithm for mlap-d}

In this section we present our online algorithm for {\MLAPD} with 
competitive ratio at most $D^22^D$.
To this end, we will give an online algorithm that achieves
competitive ratio $R_L = (2+1/L)^{D-1}$ for $L$-decreasing trees.
Taking $L = D/2$ and using the reduction to $L$-decreasing trees from
Theorem~\ref{thm:reduction}, we obtain a
$D^22^D$-competitive algorithm for arbitrary trees.


\subsection{Intuitions} 

Consider the optimal $2$-competitive algorithm for {\MLAPD} for
trees of depth~$2$~\cite{jrp-soda-2014}. Assume that the tree is $L$-decreasing, 
for some large $L$. (Thus $\length_q \gg \length_v$, for each leaf $v$.)
Whenever a pending request
reaches its deadline, this algorithm serves a subtree $X$ consisting
of $r,q$ and the set of leaves with the earliest deadlines and total
weight of about $\length_q$.  This is a~natural strategy: We
have to pay at least $\length_q$ to serve the expiring request, so
including an additional set of leaves of total weight $\length_q$ can
at most double our overall cost. But, assuming that no new requests arrive,
serving this $X$ can significantly reduce the cost in the future,
since servicing these leaves individually is expensive: it would
cost $\length_v+\length_q$ per each leaf $v$, compared to the
incremental cost of $\length_v$ to include $v$ in $X$.

For $L$-decreasing trees with three levels (that is, for $D=3$), 
we may try to iterate this idea. 
When constructing a service tree $X$, we start by adding to $X$
the set of most urgent children of $q$ whose total weight is roughly
$\length_q$.  Now, when choosing nodes of depth $3$, we have two
possibilities: (1) for each $v\in X-\braced{r,q}$ we can add to $X$
its most urgent children of combined weight $\length_v$ (note that
their total weight will add up to roughly $\length_q$, because of the
$L$-decreasing property), or (2) from the set of \emph{all} children
of the nodes in $X-\braced{r,q}$, add to $X$ the set of total weight
roughly $\length_q$ consisting of (globally) most urgent children. 

It is not hard to show that option (1) does not lead to a
constant-competitive algorithm: The counter-example involves an
instance with one node $w$ of depth $2$ having many children with
requests with early deadlines and all other leaves having requests
with very distant deadlines. Assume that $\length_q=L^2$,
$\length_w=L$, and that each leaf has weight $1$. 
The example forces the
algorithm to serve the children of $w$ in small batches of size $L$
with cost more than $L^2$ per batch or $L$ per each child of $w$, while the
optimum can serve all the requests in the children of $w$ at once with cost $O(1)$ per
request, giving a lower bound $\Omega(L)$ on the competitive ratio.
(The requests at other nodes can be ignored in the optimal
solution, as we can keep repeating the above strategy in a manner
similar to the lower-bound technique for {\SPMLAP}
that will be described in Section~\ref{sec: one-phase MLAP}.
Reissuing requests at the 
nodes other than $w$ will not increase the cost of the optimum.) 
A more intricate example
shows that option (2) by itself is not sufficient to guarantee
constant competitiveness either.

The idea behind our algorithm, for trees of depth $D=3$, is to do 
\emph{both} (1) and (2) to obtain $X$. This increases the cost of each
service by a constant factor, but it protects the algorithm against
both bad instances. The extension of our algorithm to depths 
$D>3$ carefully iterates the
process of constructing the service tree $X$, to ensure that for each
node $v\in X$ and for each level $i$ below $v$ we add to $X$
sufficiently many urgent descendants of $v$ at that level.


\subsection{Notations}

To give a formal description, we need some more notations.
For any set of nodes $Z\subseteq\calT$, let $Z^i$ denote the set of
nodes in $Z$ of depth $i$ in tree $\calT$. (Recall
that $r$ has depth $0$, $q$ has depth $1$, and leaves have depth at
most $D$.) Let also $Z^{<i} =\bigcup_{j=0}^{i-1}Z^j$ and $Z^{\le i}
=Z^{<i}\cup Z^i$. These notations can be combined with the notation
$Z_x$, so, e.g., $Z_x^{<i}$ is the set of all descendants of $x$ that
belong to $Z$ and whose depth in $\calT$ is smaller than $i$.

We assume that all the deadlines in the given instance are
distinct. This may be done without loss of generality, as in case of
ties we can modify the deadlines by infinitesimally small
perturbations and obtain an algorithm for the general case. 

At any given time $t$ during the computation of the algorithm, for
each node $v$, let $d^t(v)$ denote the earliest deadline among all
requests in $\calT_v$ (i.e., among all descendants of $v$) that are
pending for the algorithm; if there is no pending request in
$\calT_v$, we set $d^t(v)=+\infty$.  We will use the function $d^t$ as
the urgency (see Section~\ref{sec: preliminaries})
of nodes at time $t$, i.e., a node $u$ will be
considered more urgent than a node $v$ if $d^t(u)<d^t(v)$.


\subsection{Algorithm~{\OnAlgTreesDeadlines}}

At any time $t$ when some request expires, that is when  $t=d^t(r)$, the
algorithm serves a~subtree $X$ constructed by first initializing
$X = \braced{r,q}$, and then incrementally augmenting $X$
according to the following pseudo-code:

\begin{tabbing}
aaa \= aaa \= aaa \= aaa \= aaa \= aaa \= \kill
\> \textbf{for each} depth $i=2,\ldots,D$\\
\> \> $Z^i\gets$ set of all children of nodes in $X^{i-1}$ \\
\> \> \textbf{for each} $v\in X^{<i}$ \\
\> \> \> $U(v,i,t) \gets \urgentnodes(Z^i_v,\length_v,d^t)$ \\
\> \> \> $X \gets X \cup U(v,i,t)$
\end{tabbing}

In other words, at depth $i$, we restrict our attention to $Z^i$, the
children of all the nodes in $X^{i-1}$, i.e., of the nodes that we
have previously selected to $X$ at level $i-1$. (We start with $i=2$
and $X^1=\braced{q}$.) Then we iterate over all $v\in X^{<i}$ and
we add to $X$ the set $U(v,i,t)$ of nodes from $\calT^i_v$
(descendants of $v$ at depth $i$) whose parents are in $X$, one by
one, in the order of increasing deadlines, stopping when either their
total weight exceeds $\length_v$ or when we run out of such
nodes. Note that these sets do not need to be disjoint.

The constructed set $X$ is a service tree, as we are adding to it only
nodes that are children of the nodes already in $X$.

Let $\rho$ be the request triggering the service at time $t$, i.e.,
satisfying $d_\rho=t$.  (By the assumption about different deadlines,
$\rho$ is unique.)  Naturally, all the nodes $u$ on the path from $r$
to $\trignode_\rho$ have $d^t(u)=t$ and qualify as the most urgent,
thus the node $\trignode_\rho$ is included in $X$. Therefore every
request is served before its deadline.


\subsection{Analysis}

Intuitively, it should be clear that Algorithm~{\OnAlgTreesDeadlines}
cannot have a better competitive ratio than $\length(X)/\length_q$: If all
requests are in $q$, the optimum will serve only $q$, while our
algorithm uses a set $X$ with many nodes that turn out to be
useless. As we will show, via an iterative charging argument, the ratio
$\length(X)/\length_q$ is actually achieved by the algorithm.

Recall that $R_L = (2+1/L)^{D-1}$.
We now prove a bound on the cost of the service tree.

\begin{lemma}\label{l:d:complete} 
Let $X$ be the service tree produced by
Algorithm~{\OnAlgTreesDeadlines} at time $t$.  Then $\length(X)\leq
R_L\cdot\length_q$.
\end{lemma}

\begin{proof}
We prove by induction that $\length(X^{\leq i}) \leq
(2+1/L)^{i-1}\length_q$ for all $i\leq D$.

The base case of $i=1$ is trivial, as $X^{\le 1}=\braced{r,q}$ and
$\length_r=0$.  For $i\geq 2$, $X^i$ is the~union of the sets $U(v,i,t)$
over all nodes $v \in X^{<i}$. Since $\calT$ is $L$-decreasing, each
node in the set $U(v,i,t)$ has weight at most $\length_v/L$.  Thus
the total weight of $U(v,i,t)$ is at most
$\length(U(v,i,t)) \le \length_v + \length_v/L = (1+1/L)\length_v$.
Therefore, by the inductive assumption, we get that
\begin{align*}
\length(X^{\leq i}) &\leq (1+(1+1/L)) \cdot \length(X^{<i}) 
\\
&\leq (2+1/L)\cdot(2+1/L)^{i-2}\length_q 
= (2+1/L)^{i-1}\length_q\,,
\end{align*}
proving the induction step and completing the proof that 
$\length(X)\leq R_L\cdot\length_q$.
\end{proof}

The competitive analysis uses a charging scheme. Fix some optimal
schedule $\optschedS$. Consider a service $(X,t)$ of
Algorithm~{\OnAlgTreesDeadlines}.  We will identify in $X$ a subset of
``critically overdue'' nodes (to be defined shortly) of total weight
at least $\length_q\geq\length(X)/R_L$, and we will show that for each
such critically overdue node $v$ we can charge the portion $\length_v$
of the service cost of $X$ to an earlier service in $\optschedS$ that
contains $v$.  Further, any node in service of $\optschedS$ will be
charged at most once.  This implies that the total cost of our
algorithm is at most $R_L$ times the optimal cost, giving us an upper
bound of $R_L$ on the competitive ratio for $L$-decreasing trees.

In the proof, by $\advtrans^t_v$ we denote the time of the first
service in $\optschedS$
that includes $v$ and is strictly after time $t$; we
also let $\advtrans^t_v=+\infty$ if no such service exists  
($\advtrans$ stands for \emph{next optimal service}).
For a service $(X,t)$ of the algorithm, we say that a node $v\in X$ is
\emph{overdue} at time $t$ if $d^t(v)<\advtrans^t_v$. Servicing of
such $v$ is delayed in comparison to $\optschedS$, because
$\optschedS$ must have served $v$ before or at time $t$.  Note
also that $r$ and $q$ are overdue at time $t$, as $d^t(r)=d^t(q)=t$ by
the choice of the service time.  We define $v\in X$ to be \emph{critically
  overdue} at time $t$ if (i) $v$ is overdue at $t$, and (ii)
there is no other service of the algorithm in the time interval
$(t,\advtrans^t_v)$ in which $v$ is overdue.

We are now ready to define the charging for a service $(X,t)$.  For
each $v\in X$ that is critically overdue, we charge its weight
$\length_v$ to the last service of $v$ in $\optschedS$ before or at
time $t$.  This charging is well-defined as, for each overdue $v$,
there must exist a previous service of $v$ in $\optschedS$. The
charging is obviously one-to-one because between any two services in
$\optschedS$ that involve $v$ there may be at most one service of the
algorithm in which $v$ is critically overdue.  The following lemma
shows that the total charge from $X$ is large enough.


\begin{lemma}\label{l:d:charge}
Let $(X,t)$ be a service of Algorithm~{\OnAlgTreesDeadlines} and
suppose that $v\in X$ is overdue at time $t$. 
Then the total weight of critically overdue nodes in $X_v$ at time $t$
is at least~$\length_v$.
\end{lemma}

\begin{proof}
The proof is by induction on the depth of $\calT_v$, the induced
subtree rooted at $v$. 

The base case is when $\calT_v$ has depth $0$, that is when $v$ is a
leaf.  We show that in this case $v$ must be critically overdue, which
implies the conclusion of the lemma.  Towards contradiction, suppose
that there is some other service at time $t'\in (t , \advtrans^t_v)$
in which $v$ is overdue.  Since $v$ is a leaf, after the service at
time $t$ there are no pending requests in $\calT_v =\braced{v}$.  This
would imply that there is a request $\rho$ with $\trignode_\rho = v$
such that $t < a_\rho \le d_\rho < \advtrans^t_v$.  But this is not
possible, because $\optschedS$ does not serve $v$ in the time
interval $(t,\advtrans^t_v)$.  Thus $v$ is critically overdue and the
base case holds.

Assume now that $v$ is not a leaf, and that the lemma holds for all
descendants of $v$.  If $v$ is critically overdue, the conclusion of
the lemma holds.

Thus we can now assume that $v$ is not critically overdue. This means
that there is a~service $(Y,t')$ of Algorithm~{\OnAlgTreesDeadlines}
with $t<t'<\advtrans^t_v$ which contains $v$ and such that $v$ is
overdue at $t'$.  Thus $\advtrans^t_v=\advtrans^{t'}_v$.

Let $\rho$ be the request with $d_\rho=d^{t'}(v)$, i.e., the most
urgent request in $\calT_v$ at time $t'$.


\begin{figure}
\begin{center}
\includegraphics[height=2.3in]{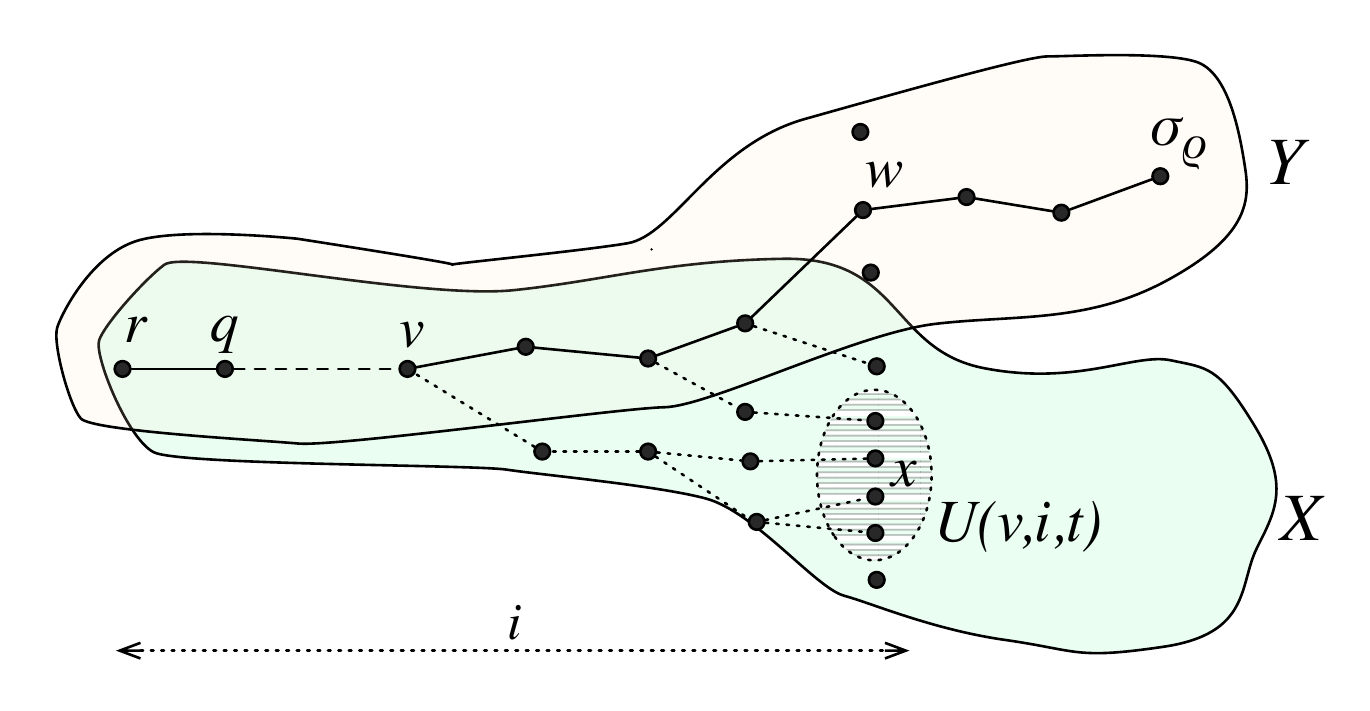}
\caption{Illustration of the proof of Lemma~\ref{l:d:charge}.}
\label{fig: online mlapd analysis}
\end{center}
\end{figure}

We claim that $a_\rho \leq t$, i.e., $\rho$ arrived no later than
at time $t$. Indeed, since $v$ is overdue at time $t'$, it follows
that $d_\rho < \advtrans^{t'}_v=\advtrans^t_v$.  The optimal schedule
$\optschedS$ cannot serve $\rho$ after time $t$, as $\optschedS$ has no
service from $v$ in the interval $(t,d_\rho]$. Thus $\optschedS$ must
have served $\rho$ before or at $t$, and hence $a_\rho \le t$, as
claimed.

Now consider the path from $\trignode_\rho$ to $v$ in $Y$.  (See
Figure~\ref{fig: online mlapd analysis}.)  As $\rho$ is pending
for the algorithm at time $t$ and $\rho$ is not served by $(X,t)$,
it follows that $\trignode_\rho\not\in X$.  Let $w$ be the last node
on this path in $Y-X$. Then $w$ is well-defined and $w\neq v$, as
$v\in X$. Let $i$ be the depth of $w$.  Note that the parent of $w$ is
in $X^{<i}_v$, so $w\in Z^i$ in the algorithm when $X$ is constructed.

The node $\trignode_\rho$ is in $\calT_w$ and $\rho$ is pending at
$t$, thus we have $d^t(w)\leq d_\rho$.  Since $w \in Z^i$ but $w$ was
not added to $X$ at time $t$, we have that $\length(U(v,i,t))\geq
\length_v$ and each $x\in U(v,i,t)$ is at least as urgent as $w$.
This implies that such $x$ satisfies
\begin{equation*}
 d^t(x)\leq d^t(w) \le d_\rho<\advtrans^{t'}_v=\advtrans^t_v\leq\advtrans^t_x,
\end{equation*}
thus $x$ is overdue at time $t$. By the inductive assumption, the
total weight of critically overdue nodes in each induced 
subtree $X_x$ is at
least $\length_x$. Adding these weights over all $x\in U(v,i,t)$,
we obtain that the total weight of critically overdue nodes in $X_v$ is at
least $\length( U(v,i,t))\geq \length_v$, completing the proof.
\end{proof}

Now consider a service $(X,t)$ of the algorithm. The quasi-root $q$ is
overdue at time $t$, so Lemmata~\ref{l:d:charge} and~\ref{l:d:complete}
imply that the charge from $(X,t)$ is at least
$\length_q\ge \length(X)/R_L$. Since
each node in any service in $\optschedS$ is charged at most once, we
conclude that Algorithm~{\OnAlgTreesDeadlines} is $R_L$-competitive for any
$L$-decreasing tree $\calT$.

From the previous paragraph, using Theorem~\ref{thm:reduction}, we now
obtain that there exists a $DLR_L = DL(2+1/L)^{D-1}$-competitive
algorithm for general trees. For $D\ge 2$, choosing $L=D/2$ yields a competitive
ratio bounded by
$
\half D^2 2^{D-1}\cdot  (1+1/D)^D
		\leq \onefourth D^22^D\cdot \e
		\leq D^22^D
$.
(For $D=1$ there is a trivial $1$-competitive algorithm for {\MLAPD}.)
Summarizing, we obtain the following result. 

\begin{theorem}
There exists a $D^22^D$-competitive online algorithm for {\MLAPD}.
\end{theorem}


\section{A Competitive Algorithm for {\MLAP}}
\label{sec: competitive algorithm for mlap}

In this section we show that there is an online algorithm for {\MLAP} 
whose competitive ratio for trees of depth $D$ is $O(D^42^D)$. 
As in Section~\ref{sec: competitive algorithm for mlap-d}, we will
assume that the tree $\calT$ in the instance is $L$-decreasing.
Then, for $L$-decreasing trees, we will present a competitive algorithm,
which will imply the existence of a competitive algorithm for
arbitrary trees by using Theorem~\ref{thm:reduction}
and choosing an appropriate value of $L$.


\subsection{Preliminaries and Notations}
\label{subsec: mlap notations}

Recall that $\omega_\rho(t)$ denotes the waiting cost function of a request
$\rho$. As explained in Section~\ref{sec: preliminaries}, we assume that 
the waiting cost functions are continuous. 
(In Section~\ref{sec: general waiting costs} we discuss how to extend our
results to arbitrary waiting cost functions.)
We will overload this notation, so that we can talk about the waiting
cost of a set of requests or a set of
nodes. Specifically, for a set $P$ of requests and a set $Z$ of
nodes, let
\begin{equation*}
\omega_P(Z,t) = \sum_{\rho\in P: \trignode_\rho\in Z} \omega_\rho(t).
\end{equation*}
Thus $\omega_P(Z,t)$ is the total waiting cost of the requests from
$P$ that are issued in $Z$. 
We sometimes omit $P$, in which case the notation refers to the set of
all requests in the instance, that is $\omega(Z,t) = \omega_{\calR}(Z,t)$. 
Similarly, we omit $Z$ when $Z$ contains all nodes,
that is $\omega_P(t) = \omega_P(\calT,t)$.


\paragraph{Maturity time.}

In our algorithm for {\MLAPD} in Section~\ref{sec: competitive algorithm for mlap-d}, 
the times of services and the
urgency of nodes are both naturally determined by the deadlines. For
{\MLAP} with continuous waiting costs there are no hard deadlines. 
Nevertheless, we can still introduce the notion of \emph{maturity
  time} of a node, which is, roughly speaking, the time when some 
subtree rooted at this node has its waiting cost equal to its
service cost; this subtree is then called \emph{mature}.
This maturity time will be our urgency function, as discussed 
earlier in Section~\ref{sec: preliminaries}. We use the maturity
time in two ways: one, the maturity times of the 
quasi-root determine the service times, and two, maturity times
of other nodes are used to prioritize them for inclusion in
the service trees. We now proceed to define these notions.

Consider some time $t$ and any set $P\subseteq\calR$ of requests. A
subtree $Z$ of $\calT$ (not necessarily rooted at $r$) is called
\emph{$P$-mature} at time $t$ if $\omega_P(Z,t) \ge \length(Z)$.
Also, let $\treematurity_P(Z)$ denote the minimal time $\tau$ such
that $\omega_P(Z,\tau) = \length(Z)$; we let $\treematurity_P(Z) =
\infty$ if such $\tau$ does not exist. In other words,
$\treematurity_P(Z)$ is the earliest $\tau$ at which $Z$ is
$P$-mature. Since $\omega_P(Z,0)=0$ and $\omega_P(Z,t)$ is a
non-decreasing and continuous function of $t$, $\treematurity_P(Z)$ is
well-defined.

For a node $v$, let the \emph{$P$-maturity time of $v$}, denoted
$\vertmaturity_P(v)$, be the minimum of values $\treematurity_P(Z)$
over all subtrees $Z$ of $\calT$ rooted at $v$.  The tree $Z$ that
achieves this minimum will be denoted $C_P(v)$ and called the
\emph{$P$-critical subtree rooted at $v$}; if there are more such trees,
choose one arbitrarily. Therefore we have
$\omega_P(C_P(v),\vertmaturity_P(v))=\length(C_P(v))$.

The following simple lemma guarantees that the maturity time of any
node in the $P$-critical subtree $C_P(v)$ is upper bounded by the maturity time of $v$.


\begin{lemma}
\label{lem: maturity-critical}
Let $u\in C_P(v)$ and let $Y =(C_P(v))_u$ be the induced subtree of 
$C_P(v)$ rooted at $u$. Then 
$\vertmaturity_P(u)\leq \treematurity_P(Y)\leq \vertmaturity_P(v)$.
\end{lemma}

\begin{proof}
The first inequality follows directly from the definition of $\vertmaturity_P(u)$. 
To show the second inequality, we proceed by contradiction. 
Let $t = \vertmaturity_P(v)$.
If the second inequality does not
hold, then $u\neq v$ and $\omega_P(Y,t)< \length(Y)$. 
Take $Y' = C_P(v) - Y$, which is a tree rooted at $v$.
Since $\omega_P(C_P(v),t)=\length(C_P(v))$, we have that 
$\omega_P(Y',t) = \omega_P(C_P(v),t) - \omega_P(Y,t)
	> \length(C_P(v)) - \length(Y)
	= \length(Y')$. 
This in turn implies that $\treematurity_P(Y') < t$, which is a
contradiction with the definition of $t = \vertmaturity_P(v)$.
\end{proof}

Most of the references to maturity of a node or to its critical set
will be made with respect to the set of requests pending for our
algorithm at a given time.  For any time $t$, we will use notation
$\vertmaturity^t(v)$ and $C^t(v)$ to denote the time
$\vertmaturity_P(v)$ and the $P$-critical subtree $C_P(v)$, where $P$
is the set of requests pending for the algorithm at time $t$; if the
algorithm schedules a service at some time $t$, $P$ is the set of
requests that are pending at time $t$ \emph{right before} the service
is executed. Note that in general it is possible that
$\vertmaturity^t(v)<t$. However, our algorithm will maintain the
invariant that for the quasi-root $q$ we will have
$\vertmaturity^t(q)\geq t$ at each time $t$.


\subsection{Algorithm}

We now describe our algorithm for $L$-decreasing trees.
A service will occur at each maturity
time of the quasi-root $q$ (with respect to the pending requests), 
that is at each time $t$ for which $t=\vertmaturity^t(q)$.
At such a time, the algorithm chooses a service that contains the critical subtree
$C = C^t(q)$ of $q$ and an extra set $E$, whose service cost is not much more
expensive than that of $C$. The extra set is constructed similarly as
in Algorithm~{\OnAlgTreesDeadlines}, where the urgency of nodes is now measured
by their maturity time. In other words, our urgency function
is now $f = \vertmaturity^t$ (see Section~\ref{sec: preliminaries}.)
As before, this extra set will be a union of a
system of sets $U(v,i,t)$ for $i=2,\ldots,D,$ and $v\in C^{<i}\cup
E^{<i}$, except that now, for technical reasons,
the sets $U(v,i,t)$ will be mutually disjoint and also disjoint from $C$.

\paragraph{Algorithm~{\OnAlgTreesGeneral}.}

At any time $t$ such that $t=\vertmaturity^t(q)$, serve the set 
$X = C\cup E$ constructed according to the following pseudo-code:
\begin{tabbing}
aaa \= aaa \= aaa \= aaa \= aaa \= aaa \= \kill
\>
$C\gets C^t(q) \cup \braced{r}$
\\
\> $E\gets \emptyset$
\\
\> \textbf{for each} depth $i=2,\ldots,D$
\\
\> \> $Z^i\gets$ set of all nodes in $\calT^i - C$ whose parent is in $C\cup
E$
\\
\> \> \textbf{for each} $v\in (C\cup E)^{<i}$ 
\\
\> \> \> $U(v,i,t) \gets \urgentnodes(Z^i_v,\length_v,\vertmaturity^t)$ 
\\
\> \> \> $E \gets E \cup U(v,i,t)$
\\
\> \> \> $Z^i \gets Z^i - U(v,i,t)$
\end{tabbing}
At the end of the instance (when $t=H$, the time horizon), if there
are any pending requests, issue the last service that contains
all nodes $v$ with a pending request in $\calT_v$.


\medskip

Note that $X=C\cup E$ is indeed a service tree, as it contains $r,q$ and we are
adding to it only nodes  that are children of the nodes already in $X$.
The initial choice and further changes of $Z^i$ imply that the sets
$U(v,i,t)$ are pairwise disjoint and disjoint from $C$ -- a fact that will
be useful in our analysis. 


We also need the following fact.

\begin{lemma}\label{lem: maturity times of q}
	(a) Suppose that Algorithm~{\OnAlgTreesGeneral} issues a service at
	a time $t$, that is $\vertmaturity^{t}(q) = t$.
	Denote by $\vertmaturity^{t^+}(q)$ the maturity time of $q$ right
	after the service at time $t$. Then $\vertmaturity^{t^+}(q) > t$.
(b) At any time $t$ we have	$\vertmaturity^t(q)\geq t$.
\end{lemma}

To clarify the meaning of ``right after the service''
in this lemma, $\vertmaturity^{t^+}(q)$ is defined formally as the 
limit of $\vertmaturity^{\tau}(q)$, with $\tau$ approaching $t$ from the right.

\begin{proof}
(a) 
Let $\vertmaturity^t(q) = t$ and let $(X,t)$ be the service at time
$t$. This means that we have $\omega(X,t) = \ell(X)$ and
$\omega(Y,t)\leq\ell(Y)$ for all subtrees $Y$ of $\calT$ rooted at
$r$.  Consider any subtree $Y$ of $\calT$ rooted at $r$ different from
$X$.  Denoting by $\omega(Y,t^+)$ the waiting cost of the packets that
are pending in $Y$ right after the service $(X,t)$, it is sufficient
to prove that $\omega(Y,t^+) < \length(Y)$.

Towards contradiction, suppose that $\omega(Y,t^+) \ge \length(Y)$.
Then we have
\begin{align*}
\omega(X\cup Y,t) &= \omega(X,t) + \omega(Y-X,t)
	\\
			&= \omega(X,t) + \omega(Y,t^+)
	\\
			&\ge \length(X) + \length(Y)
	\\
			&> \length(X\cup Y),
\end{align*}
where the last (strict) inequality follows from $q\in X\cap Y$ and
$\length_q > 0$.
But $X\cup Y$ is a subtree of $\calT$ rooted at $r$, so
the inequality $\omega(X\cup Y,t) > \length(X\cup Y)$
contradicts our assumption that $\vertmaturity^t(q) = t$.	

(b)
The lemma holds trivially at the beginning, at time $t=0$. In any time interval
without new requests released nor services, the inequality
$\vertmaturity^t(q)\geq t$ is preserved, by the definition of
the service times and continuity of waiting cost functions.
Releasing a request $\rho$ at a time $a_\rho = t$ 
cannot decrease $\vertmaturity^t(q)$ to below $t$, because the waiting cost function
of $\rho$ is identically 0 up to $t$ and thus
releasing $\rho$ does not change the waiting costs at time $t$ or before. 
Finally, part~(a) implies that the inequality is also preserved when
services are issued.
\end{proof}

By Lemma~\ref{lem: maturity times of q} (and the paragraph before),
the definition of the algorithm is sound, that is the sequence of
service times is non-decreasing. In fact, the lemma
shows that no two services can occur at the same time.
 

\subsection{Competitive Analysis}

We now present the proof of the existence of an
$O(D^42^D)$-competitive algorithm for $\MLAP$ for
trees of depth~$D$. The overall argument is quite intricate, so we 
will start by summarizing its main steps:

\begin{itemize}

\item First, as explained earlier, we will assume that the tree
$\calT$ in the instance is $L$-decreasing. 
For such $\calT$ we will show that Algorithm~{\OnAlgTreesGeneral}
has competitive ratio $O(D^2R_L)$, where
$R_L = (2+1/L)^{D-1}$. Our bound on the competitive ratio for
arbitrary trees will then follow, by using Theorem~\ref{thm:reduction}
and choosing an appropriate value of $L$ (see Theorem~\ref{thm: general mlap competitive}).

\item For $L$-decreasing trees,
	the bound of the competitive ratio of Algorithm~{\OnAlgTreesGeneral}
	involves four ingredients:
	
	\begin{itemize}
		
		\item We show (in Lemma~\ref{lem: mlap, cost <= 2*service})
		that the total cost of Algorithm~{\OnAlgTreesGeneral}
			is at most twice its service cost.
			
		\item Next, we show that the service cost of
			Algorithm~{\OnAlgTreesGeneral} can be bounded 
			(within a constant factor) by the total cost of all
			critical subtrees $C^t(q)$ of the service trees in its schedule.
			
		\item To facilitate the estimate of the adversary cost, 
			we introduce the concept of a \emph{pseudo-schedule}
			denoted $\pseudoschedS$. The pseudo-schedule
                        $\pseudoschedS$ is a collection of 
			\emph{pseudo-services}, which include the services
			from the original adversary schedule $\optschedS$.
			We show (in Lemma~\ref{lem: bound on pseudo sched})
			that the adversary pseudo-schedule has service
			cost not larger than $D$ times the cost of $\optschedS$.
			Using the pseudo-schedule allows us to ignore the waiting cost in the
			adversary's schedule.
			
		\item With the above bounds established, it remains to
			show that the total cost of critical subtrees in the
			schedule of Algorithm~{\OnAlgTreesGeneral}
			is within a constant factor of the service cost of the
			adversary's	pseudo-schedule. This is accomplished through a 
			charging scheme that charges nodes (or, more precisely,
			their weights) from each critical subtree of
			Algorithm~{\OnAlgTreesGeneral} to their appearances
			in some earlier adversary pseudo-services.
		
	\end{itemize}
	
\end{itemize}


\paragraph{Two auxiliary bounds.}
We now assume that $\calT$ is $L$-decreasing and proceed with our
proof, according to the outline above.

The definition of the
maturity time implies that the waiting cost of all the requests
served is at most the service cost $\length(X)$, as otherwise
$X$ would be a good candidate for a critical subtree at some
earlier time. Denoting by $\schedS$ the schedule computed by
Algorithm~{\OnAlgTreesGeneral}, we thus obtain:


\begin{lemma}\label{lem: mlap, cost <= 2*service}
$\cost(\schedS) \le 2\cdot\scost(\schedS)$.
\end{lemma}

Using Lemma~\ref{lem: mlap, cost <= 2*service},
 we can restrict ourselves to bounding the service
cost, losing at most a factor of $2$.  We now bound the cost of 
a given service $X$; recall that $R_L = (2+1/L)^{D-1}$.


\begin{lemma}\label{l:complete}
Each service tree $X=C\cup E$ constructed by the algorithm satisfies
$\length(X)\leq R_L\cdot\length(C)$.
\end{lemma}

\begin{proof}
Since $\calT$ is $L$-decreasing, the weight of each node that is a
descendant of $v$ is at most $\length_v/L$ and thus
$\length(U(v,i,t))\leq (1+1/L)\length_v$.

We now estimate $\length(X)$. We claim and prove by induction for
$i=1,\ldots,D$ that
\begin{equation}
\length(X^{\le i})\leq(2+1/L)^{i-1}\length(C^{\le i})\,.
	\label{eqn: complete mlap main ineq}
\end{equation}
The base case for $i=1$ is trivial, as $X^{\le 1}=C^{\le
  1}=\braced{r,q}$. For $i\geq 2$, the set
$X^i$ consists of $C^i$ and the sets $U(v,i,t)$,
for $v \in X^{<i}$. Each of these sets $U(v,i,t)$ has
weight at most $(1+1/L)\length_v$. Therefore
\begin{equation}
\length(X^i)\leq (1+1/L)\length(X^{<i})+\length(C^i)\,.
	\label{eqn: complete mlap aux ineq}
\end{equation}
Now, using (\ref{eqn: complete mlap aux ineq}) and the
inductive assumption (\ref{eqn: complete mlap main ineq}) for $i-1$,
we get
\begin{align*}
\length(X^{\le i})&= \length(X^{<i})+ \length(X^i)
			\\
			&\leq (2+1/L)\length(X^{<i})+\length(C^i)
			\\
			&\leq (2+1/L)^{i-1}\length(C^{<i})+\length(C^i)
			\,\leq\, (2+1/L)^{i-1}\length(C^{\le i}).
\end{align*}
Taking $i=D$ in~(\ref{eqn: complete mlap main ineq}), the lemma follows.
\end{proof}

 
\paragraph{Waiting costs and pseudo-schedules.}

Our plan is to charge the cost of Algorithm~{\OnAlgTreesGeneral} to
the optimal (or the adversary's) cost. Let $\optschedS$ be an optimal
schedule.  To simplify this charging, we extend $\optschedS$ by adding
to it pseudo-services, where a \emph{pseudo-service from a node $v$}
is a partial service of cost $\length_v$ that consists only of the
edge from $v$ to its parent.  We denote this modified schedule
$\pseudoschedS$ and call it a \emph{pseudo-schedule}, reflecting the
fact that its pseudo-services are not necessarily subtrees of $\calT$
rooted at $r$. Adding such pseudo-services will allow us to
ignore the waiting costs in the optimal schedule.

We now define more precisely how to obtain $\pseudoschedS$ from $\optschedS$.
For each node $v$ independently we define the times when new
pseudo-services of $v$ occur in $\pseudoschedS$. Intuitively, we introduce
these pseudo-services at intervals such that the waiting cost
of the requests that arrive in $\calT_v$ during these intervals adds
up to $\length_v$.  The formal description of this process is given in
the pseudo-code below, where we use notation $\calR(>t)$ for the set of
requests $\rho\in\calR$ with $a_\rho>t$ (i.e., requests issued after
time $t$). Recall that $H$ denotes the time horizon.

\begin{tabbing}
aaa \= aaa \= aaa \= aaa \= aaa \= aaa \= \kill
\>
$t\gets -\infty$
\\
\> \textbf{while} $\omega_{\calR(>t)}(\calT_v,H)\ge\length_v$
\\
\> \>
let $\tau$ be the earliest time such that
$\omega_{\calR(>t)}(\calT_v,\tau)=\length_v$
\\
\> \>
add to $\pseudoschedS$ a pseudo-service of $v$ at $\tau$ 
\\
\> \>
$t\gets \tau$ 
\end{tabbing}

We apply the above procedure to all the nodes $v\in\calT-\{r\}$
such that $\calR$ contains a request in $\calT_v$. The new
pseudo-schedule $\pseudoschedS$ contains all the services of
$\optschedS$ (treated as sets of pseudo-services of all served nodes)
and the new pseudo-services added as above.  The service cost of the
pseudo-schedule, $\scost(\pseudoschedS)$, is defined naturally as the
total weight of the nodes in all its pseudo-services.


\begin{lemma}\label{lem: bound on pseudo sched}
$\scost(\pseudoschedS) \leq D\cdot\cost(\optschedS)$.
\end{lemma}

\begin{proof}
It is sufficient to show that the total service cost of the new
pseudo-services added inside the while loop is at most
$\scost(\optschedS) + D\cdot \wcost(\optschedS)$: Adding
$\scost(\optschedS)$ once more to account for the service cost of the
services of $\optschedS$ that are included in $\pseudoschedS$, and
using our assumption that $D\ge 3$, we obtain $\scost(\pseudoschedS)
\leq 2\cdot\scost(\optschedS) + D\cdot \wcost(\optschedS)\leq
D\cdot\cost(\optschedS)$, thus the lemma follows.

To prove the claim, consider some node $v$, and a pair of times
$t,\tau$ from one iteration of the while loop, when a new
pseudo-service was added to $\pseudoschedS$ at time $\tau$.  This
pseudo-service has cost $\length_v$.  In $\optschedS$, either there is
a service in $(t,\tau]$ including $v$, or the total waiting cost of the
requests within $\calT_v$ released in this interval is equal to
$\omega_{\calR(>t)}(\calT_v,\tau)=\length_v$.  In the first case, we
charge the cost of $\length_v$ of this pseudo-service to any service
of $v$ in $\optschedS$ in $(t,\tau]$.  Since we consider here only the
new pseudo-services, created by the above pseudo-code, this charging
will be one-to-one.  In the second case, we charge $\length_v$ to the
total waiting cost of the requests in $\calT_v$ released in the
interval $(t,\tau]$.  For each given $v$, the charges of the second
type from pseudo-services at $v$ go to disjoint sets of requests in
$\calT_v$, so each request in $\calT_v$ will receive at most one
charge from $v$.  Therefore, for each request $\rho$, its waiting cost
in $\optschedS$ will be charged at most $D$ times, namely at most once
from each node $v$ on the path from $\trignode_\rho$ to $q$.  From the
above argument, the total cost of the new pseudo-services is at most
$\scost(\optschedS) + D\cdot \wcost(\optschedS)$, as claimed.
\end{proof}

Using the bound in Lemma~\ref{lem: bound on pseudo sched}
will allow us to use $\scost(\pseudoschedS)$ as an estimate of
the optimal cost in our charging scheme, losing at most a factor of $D$ in the
competitive ratio.


\paragraph{Charging scheme.}

According to Lemma~\ref{lem: mlap, cost <= 2*service}, to establish constant
competitiveness it is sufficient
to bound only the service cost of Algorithm~{\OnAlgTreesGeneral}.
By Lemma~\ref{l:complete}
for any service tree $X$ of the algorithm we have $\length(X)\le R_L\cdot\length(C)$.  
Therefore, it is in fact sufficient to bound the total weight of the critical
sets in the algorithm's services. 
Further, using Lemma~\ref{lem: bound on pseudo sched}, instead of 
using the optimal cost in this bound, we can use the pseudo-service cost.
Following this idea, we will show how we can charge, at a constant rate,
the cost of
all critical sets $C$ in the algorithm's services to the adversary
pseudo-services. 

The basic idea of our charging method is similar to that for {\MLAPD}.
The argument in Section~\ref{sec: competitive algorithm for mlap-d}
can be interpreted as an iterative charging scheme, where we have a
charge of $\length_q$ that originates from $q$, and this charge is
gradually distributed and transferred down the service tree, through
overdue nodes, until it reaches critically overdue nodes that can be
charged directly to adversary services.  For {\MLAP} with general
waiting costs, the charge of $\length(C)$ will originate from the
current critical subtree $C$. Several complications arise when we attempt to
distribute the charges to nodes at deeper levels.
First,
due to gradual accumulation of waiting costs, it does not seem
possible to identify nodes in the same service tree that can be used
as either intermediate or final nodes in this process.  Instead, when
defining a charge from a node $v$, we will charge descendants of $v$ in 
\emph{earlier} services of
$v$. Specifically, the weight $\length_v$ will be charged to the set
$U(v,i,t^-)$ for some $i>\depth(v)$, where $t^-$ is the time of the previous
service of the algorithm that includes $v$.  The nodes --- or, more
precisely, services of these nodes --- that can be used as
intermediate nodes for transferring charges will be called
\emph{depth-timely}.  As before, we will argue that each charge will
eventually reach a node $u$ in some earlier service that can be
charged to some adversary pseudo-service directly.  Such service of
$u$ will be called \emph{$u$-local}, where the name reflects the property
that this service has an adversary pseudo-service of $u$ nearby (to
which its weight $\length_u$ will be charged).

\smallskip

We now formalize these notions.
Let $(X,t)$ be some service of Algorithm~{\OnAlgTreesGeneral} that includes $v$,
that is $v\in X$.
By $\prv{t}{v}$ we denote the time of the last service of $v$ before
$t$ in the schedule of the algorithm; if it does not exist, set
$\prv{t}{v}=-\infty$.
By $\nxt{t}{v}{i}$ we denote the time of the $i$th service of $v$
following $t$ in the schedule of the algorithm; if it does not exist,
set $\nxt{t}{v}{i}=+\infty$.

We say that the service of $v$ at time $t<H$ is \emph{$i$-timely}, if
$\vertmaturity^t(v)<\nxt{t}{v}{i}$; furthermore, if $v$ is
$\depth(v)$-timely, we will say simply that this service of $v$ is
\emph{depth-timely}.
We say that the service of $v$ at time $t<H$ is \emph{$v$-local}, if
this is either the first service of $v$ by the algorithm, or if
there is an adversary pseudo-service of $v$ in the interval
$(\prv{t}{v},\nxt{t}{v}{\depth(v)}]$.


Given an algorithm's service $(X,t)$, we now define the outgoing
charges from $X$.  For any $v\in X-\{r\}$, its outgoing charge is
defined as follows:
\begin{description}
	
\item{(C1)} If $t<H$ and the service of $v$ at time $t$ is both
  depth-timely and $v$-local, charge $\length_v$ to the first
  adversary pseudo-service of $v$ after time $\prv{t}{v}$.

\item{(C2)} If $t<H$ and the service of $v$ at time $t$ is 
  depth-timely but not $v$-local, charge $\length_v$ 
  to the algorithm's service at time $\prv{t}{v}$.

\item{(C3)} If $t<H$ and the service of $v$ at time $t$ is not
  depth-timely, the outgoing charge is $0$.

\item{(C4)} If $t=H$ and $v\in X$, we charge $\length_v$ to the
first adversary pseudo-service of $v$.
\end{description}

We first argue that the charging is well-defined. 
To justify (C1) suppose that this service is depth-timely and
$v$-local.  If $(X,t)$ is the first service of $v$ then $\prv{t}{v} =
-\infty$ and the charge goes to the first pseudo-service of $v$
which exists as all the requests must be served. Otherwise there
is an adversary pseudo-service of $v$ in the interval
$(\prv{t}{v},\nxt{t}{v}{\depth(v)}]$ and rule (C1) is well-defined.
For (C2), note that if the service $(X,t)$ of $v$ is not $v$-local
then there must be an earlier service including $v$.  (C3) is
trivial. For (C4), note again that an adversary transmission of $v$
must exist, as all requests must be served.

The following lemma implies that
all nodes in the critical subtree will have an outgoing charge, as needed.

\begin{lemma}\label{lem: each v in C is depth-timely}
For a transmission time $t<H$, each $v \in C^t(q)$ is 1-timely, and
thus also depth-timely.
\end{lemma}

\begin{proof}
From Lemma~\ref{lem: maturity-critical}, each $v\in C^t(q)$ satisfies 
$\vertmaturity^t(v) \leq \vertmaturity^t(q) 
					= t 
					< \nxt{t}{q}{1}
					\le \nxt{t}{v}{1}$,
where the sharp inequality follows from Lemma~\ref{lem: maturity times of q}.
\end{proof}

The following lemma captures the key property of our charging scheme.  For any
depth-timely service of $v\in X$ that is not $v$-local, it identifies
a set $U(v,i,t^-)$ in the previous service $(X^-,t^-)$ including $v$
that is suitable for receiving a charge. It is important that each such set is
used only once, has sufficient weight, and contains only depth-timely
nodes.  As we show later, these properties imply that in this charging scheme
the net charge (the difference between the outgoing and incoming charge)
from each service $X$ is at least as large as
the total weight of its critical subtree.

As in the argument for {\MLAPD}, we need to find an urgent node 
$w\in X_v$ which is not in $X^-$ and has its parent in $X^-$. There
are two important issues caused by the fact that the urgency is
given by the maturity times instead of deadlines. The first issue is
that the maturity time can decrease due to new packet arrivals --- to
handle this, we argue that if the new requests had large waiting
costs, they would guarantee the existence of a pseudo-service of node $v$ in the given
time interval and thus the algorithm's service of
$v$ would be $v$-local. The second issue is that the
maturity time is not given by a single descendant but by adding the node
contributions from the whole tree --- thus instead of searching for $w$ on
a single path, we need a more subtle, global argument to identify such $w$.


\begin{lemma}\label{l:charge}
Assume that the service of $v$ at time $t<H$ is depth-timely and
not $v$-local.  Let $i= \depth(v)$, and let $(X^-,t^-)$ be
the previous service of Algorithm~{\OnAlgTreesGeneral} including $v$, that is
$t^-=\prv{t}{v}$. Then there exists $j>i$ such that all the nodes in
the set $U(v,j,t^-)$ from the construction of $X^-$ in the
algorithm are depth-timely and $\length(U(v,j,t^-))\geq \length_v$.
\end{lemma}

\begin{proof}
Let $\matvt =\vertmaturity^t(v)$ and let $C'=C^t(v)$ be the 
critical subtree of $v$ at time $t$.
Since the service of $v$ at time $t$ is $i$-timely, we have $\matvt <\nxt{t}{v}{i}$.
(It may be the case that $\matvt <t$, but that does not
hamper our proof in any way.)
Also, since the service of $v$ at time $t$ is not $v$-local,
 it is not the first service of $v$, thus $t^-$ and
$X^-$ are defined.  

Let $P^-$ be the set of requests pending right after
time $t^-$ (including those with arrival time $t^-$ but not those
served at time $t^-$), and let $P$ be the set of requests 
with arrival time in the interval $(t^-,t]$.
The key observation is that the total waiting cost of all the requests
in $C'$ that arrived after $t^-$ satisfies
\begin{equation}
\label{eq:new}  
\omega_P(C',\matvt)<\length_v\,.
\end{equation}
To see this, simply note that $\omega_P(C',\matvt)\geq \length_v$
would imply that $\omega_{\calR(>t^-)} (\calT_v,\matvt)\geq \length_v$.
This in turn would
imply the existence of a pseudo-service of $v$ in the interval
$(t^-,\matvt]\subseteq (\prv{t}{v},\nxt{t}{v}{i}]$, which would
contradict the assumption that the service of $v$ at time $t$
is not $v$-local. (Note that if $\matvt \leq t^-$ then $\omega_P(C',\matvt)=0$
as $\matvt$ is before the arrival time of any request in $P$ and the
inequality holds trivially.)


\begin{figure}
\begin{center}
\includegraphics[height=2in]{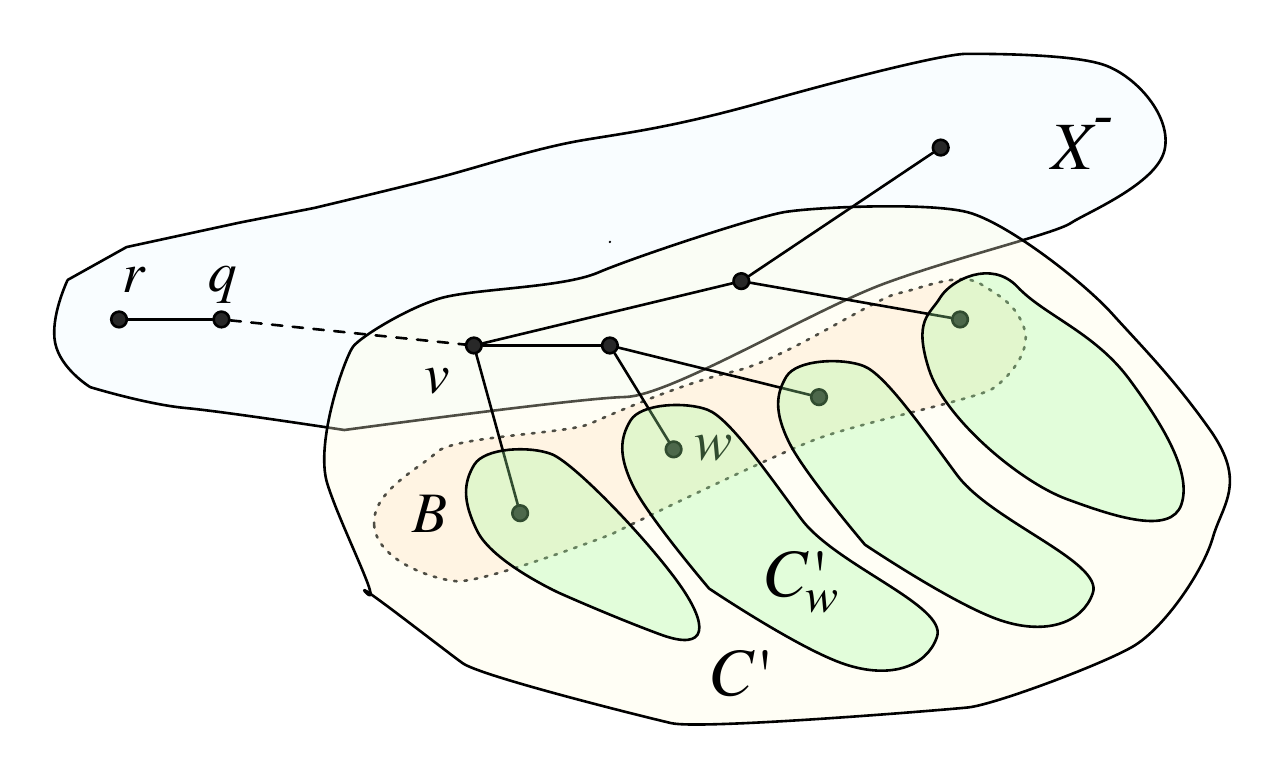}
\caption{Illustration of the proof of Lemma~\ref{l:charge}.}
\label{fig: online mlap analysis}
\end{center}
\end{figure}


Since $P^-\cup P$ contains all the requests pending at time $t$,
the choice of $\matvt$ and $C'$ implies that
\begin{equation}
\label{eq:all}  
\omega_{P^-\cup P}(C',\matvt)=\length(C')\,.
\end{equation}
$P^-$ does not contain any requests in $C'\cap X^-$, as those were served at time $t^-$;
therefore $\omega_{P^-}(C',\matvt) = \omega_{P^-}(C'-X',\matvt)$.
Letting $B$ be the set of all nodes $w\in C'-X^-$ for which $\parent(w)\in X^-$,
we have $C' - X' = \bigcup_{w\in B}C'_w$, where all sets $C'_w$, for $w\in B$,
are disjoint. (See Figure~\ref{fig: online mlap analysis}.)
Also, $v\in C'\cap X^-$. Combining these observations, and
using inequalities~(\ref{eq:new}) and~(\ref{eq:all}), we get
\begin{align*}
\textstyle \sum_{w\in B}\omega_{P^-}(C'_w, \matvt)
		&= \textstyle \omega_{P^-}(\bigcup_{w\in B}C'_w, \matvt)
\\
		&= \omega_{P^-}(C'-X',\matvt)
\\
		&= \omega_{P^-}(C', \matvt)
\\
		&= \omega_{P^-\cup P}(C',\matvt)-\omega_{P}(C',\matvt)
\\
		&> \length(C')-\length_v 
\\
		&\ge \length(C')-\length(C'\cap X^-)
\\
		&= \length(C'-X^-)
			=\textstyle \sum_{w\in B}\length(C'_w)\,.
\end{align*}
It follows that there exists $w\in B$ such that 
\begin{equation}\label{eq:old}  
\omega_{P^-}(C'_w, \matvt)>\length(C'_w)\,.
\end{equation}
Equation~(\ref{eq:old}) implies that $\vertmaturity^{t^-}(w)\leq \matvt$, 
using also the fact that $w$ was not served at $t^-$, so
$P^-$ contains exactly all the requests used to define
$\vertmaturity^{t^-}(w)$.  Let $j = \depth(w)$; note that $j>i$ as $w$
is a descendant of $v$. Since $w\not\in X^-$ but $\parent(w)\in X^-$, and
$\vertmaturity^{t^-}(w)$ is finite, the definition of the extra sets for $X^-$
implies that $U(v,j,t^-)$ has sufficient weight and all its nodes
are more urgent than $w$. More precisely, $\length(U(v,j,t^-))\geq
\length_v$ and any $z\in U(v,j,t^-)$ has $\vertmaturity^{t^-}(z)\leq
\vertmaturity^{t^-}(w)\leq \matvt$.

It remains to show that every $z\in U(v,j,t^-)$ is depth-timely at time $t^-$.
Indeed, since $\depth(z) = j \ge i+1$ and any service containing $z$ contains also $v$, we get
\begin{equation*}
\nxt{{t^-}}{z}{j}	\ge \nxt{{t^-}}{z}{i+1} 
					\ge \nxt{{t^-}}{v}{i+1}
					= \nxt{t}{v}{i} 
					> \matvt
					\geq \vertmaturity^{t^-}(z)\,,
\end{equation*}
where the last step uses the inequality $ \matvt \geq \vertmaturity^{t^-}(z)$
derived in the previous paragraph. Thus $z$ is depth-timely, as needed.
The proof of the lemma is now complete.
\end{proof}



\paragraph{Competitive analysis.} 

We are now ready to complete our competitive analysis of $\MLAP$.

\begin{theorem}\label{thm: general mlap competitive}
There exists an $O(D^42^D)$-competitive algorithm for $\MLAP$ on 
trees of depth~$D$.
\end{theorem}

\begin{proof}
We will show that Algorithm~{\OnAlgTreesGeneral}'s competitive ratio
for $L$-decreasing trees of depth $D\geq 3$
is at most $4D^2R_L$, where $R_L = (2+1/L)^{D-1}$. 
By applying Theorem~\ref{thm:reduction}, this implies that
there is an online algorithm for arbitrary trees with ratio
at most $4D^3L(2+1/L)^{D-1}$. For $L = D/2$, this
ratio is bounded by $3D^42^D$, implying the theorem (together with the
fact that for $D=1,2$, constant-competitive algorithms are known).

So now we fix an $L$-decreasing tree $\calT$ and focus our attention
on Algorithm~{\OnAlgTreesGeneral}'s schedule $\schedS$ and on the
adversary pseudo-schedule $\pseudoschedS$. Define the \emph{net charge
  from} a service $(X,t)$ in $\schedS$ to be the difference between
the outgoing and incoming charge of $(X,t)$.  Our goal is to show that
each pseudo-service in $\pseudoschedS$ is charged only a constant
number of times and that the net charge from each service $(X,t)$ in
$\schedS$ is at least $\length(X)/R_L$.

Consider first an adversary pseudo-service of $v$ at a time $\tau$. We
argue that it is charged at most $(D+3)\length_v$: If this is the
first pseudo-service of $v$, charged once from both the first service
of $v$ by rule (C1) and from the last service of $v$ at time $t=H$ by
rule (C4). In addition, by rule (C1) it may be charged $D$ times from
the last $D$ services of $v$ before $\tau$, and once from the first
service at or after $\tau$. All the charges are equal to $\length_v$.

Now consider a service $(X,t)$ of Algorithm~{\OnAlgTreesGeneral}. For
$t=H$, all the nodes of $X$ have an outgoing charge by rule (C4)
and there is no incoming charge. Thus the net charge from $X$ is
$\length(X)\geq \length(X)/R_L$.

For $t<H$, let $X=C\cup E$, where $C$ is the critical subtree and $E$
is the extra set. From Lemma~\ref{lem: each v in C is depth-timely},
all nodes in $C$ are depth-timely, so they generate outgoing charge of
at least $\length(C)$ from $X$.  Next, we show that the net charge
from the extra set $E$ is non-negative.  Recall that $E$ is a disjoint
union of sets of the form $U(w,k,t)$ and $E$ is disjoint from $C$.  If
a future service of a node $v$ generates the charge of $\length_v$
to $X$ by rule (C2), it must be the service at time $\nxt{t}{v}{1}$,
so such a charge is unique for each $v$.  Furthermore,
Lemma~\ref{l:charge} implies that one of the extra sets $U(v,j,t)$,
for $j>i$, has $\length(U(v,j,t))\geq \length_v$ and consists of
depth-timely nodes only. Thus these nodes have outgoing charges adding
up to at least $\length_v$; these charges go either to the adversary's
pseudo-services or the algorithm's services before time $t$.  We have
shown that the net charge from each extra set $U(w,k,t)$ is
non-negative; therefore, the net charge from $E$ is non-negative as
well.  We conclude that the net charge from $X$ is at least
$\length(C)$. Applying Lemma~\ref{l:complete}, we obtain that this net
charge is at least $\length(X)/R_L$.

Summing over all the services $(X,t)$ in $\schedS$, we get a bound for
the service cost of schedule $\schedS$: $\scost(\schedS) \le
(D+3)R_L\cdot\scost(\pseudoschedS)$.  Applying Lemmata~\ref{lem: mlap,
  cost <= 2*service} and~\ref{lem: bound on pseudo sched}, we get
\begin{align*}
\cost(\schedS) &\le 2\cdot \scost(\schedS)
			\\
			&\leq 2(D+3)R_L\cdot\scost(\pseudoschedS)
			\\
			&\leq 2D(D+3)R_L\cdot\cost(\optschedS)
			\;\le\;
			4D^2 R_L\cdot\cost(\optschedS).
\end{align*}
We have thus shown that Algorithm~{\OnAlgTreesGeneral}'s competitive
ratio for $L$-decreasing trees is at most $4D^2 R_L$, which, as
explained earlier, is sufficient to complete the proof.
\end{proof}


\section{Single-Phase {\MLAP}}
\label{sec: one-phase MLAP}

We now consider a restricted variant of $\MLAP$ that we refer to as
\emph{Single-Phase} {\MLAP}, or $\SPMLAP$. In $\SPMLAP$ all requests arrive at the beginning, at time $0$. 
The instance also
includes a parameter $\expiration$ representing the  common \emph{expiration time}
for all requests. We do not require that all requests are served. Any unserved
request pays only the cost of waiting until the expiration time $\expiration$. 

In the online variant of $\SPMLAP$, all requests, including their waiting cost functions,
are known to the online algorithm at time $0$. The only unknown is the expiration time $\expiration$.

Although not explicitly named, variants of $\SPMLAP$ have
been considered in~\cite{jrp-online-buchbinder,aggregation_wads_2013}, where they 
were used to show lower bounds on competitive ratios for $\MLAP$. These proofs
consist of two steps, first showing a lower bound for online $\SPMLAP$ and then arguing
that, in the online scenario, $\SPMLAP$ can be expressed as a special case of $\MLAP$.
(A corresponding property holds in the offline case as well, but is quite trivial.)
We also use the same general approach in Section~\ref{sec: mlap on paths} to show our lower bounds.

To see that (in spite of the expiration feature) $\SPMLAP$ can be thought of as a special case
of $\MLAP$, we map an instance $\calJ$ of $\SPMLAP$ into the instance
$\calJ'$ of $\MLAP$ with the property that any $R$-competitive algorithm for $\calJ'$
can be converted into an $R$-competitive algorithm for $\calJ$.
We will explain the general idea when the cost function is linear; the construction
for arbitrary cost functions is based on the same idea, but it involves some minor
technical obstacles. Let $\expiration$ be the expiration time from $\calJ$.
Choose some large integers $K$ and $M$.  The constructed instance 
$\calJ'$ consists of $K$ ``nested'' and ``compressed'' copies of $\calJ$, that we
also refer to as \emph{phases}.
In the $i$-th phase we multiply the waiting cost function of each node by $M^i$. 
We let this phase start at time $(1-M^{-i})\expiration$
(that is, at this time the requests from this phase are released) and
end at time $\expiration$. Thus the length of phase $i$ is $M^{-i}\expiration$.
The main trick is that, in $\calJ'$, at time $\expiration$ the adversary can serve
all pending requests (from all phases) at the cost that is independent of $K$, 
so the contribution of this service cost to the cost of each phase is negligibly small. 
Following this idea, any $R$-competitive
algorithm for $\calJ'$ can be converted into an $R$-competitive algorithm
for $\calJ$, except for some vanishing additive constant.
(See \cite{jrp-online-buchbinder,aggregation_wads_2013} for more details.)


\subsection{Characterizing Optimal Solutions}

Suppose that the expiration value is $\expiration = t$. Then the optimal solution is to serve some
subtree $X$ (rooted at $r$) already at time~0 and wait until the end of the phase at time $t$ with the remaining 
requests in $\barX = \calT - X$. 
So now we consider schedules of this form, that consist of one 
service subtree $X\subseteq \calT$ at time $0$.
The cost of this schedule (that we identify with $X$ itself) is
\begin{equation*}
	\cost(X,t) = \length(X) + \omega(\barX,t),
\end{equation*}
where, for any set $U\subseteq\calT$, $\omega(U,t) = \sum_{\rho}\omega_\rho(U,t)$
denotes the waiting cost of all requests in $U$ (see Section~\ref{subsec: mlap notations}.)

Our first objective
is to characterize those subtrees $X$ that are optimal for $\expiration= t$. 
This characterization will play a critical role in our online algorithm for
$\SPMLAP$, provided later in this section
and it also leads to an offline polynomial-time algorithm for computing optimal
solutions, given in Section~\ref{sec: 1p-mlap with deadlines}.

The lemma below 
can be derived by expressing $\SPMLAP$ as a linear program and using strong duality. We
provide instead a simple combinatorial proof. For each subtree $Z$ of $\calT$, we
denote its root by $r_Z$. (Also, recall that $Z_v$ is the induced
subtree of $Z$ rooted at $v$, that is, $Z_v$ contains all descendants of $v$ in $Z$.)


\begin{lemma}\label{lem: 1-phase optimal schedules}
A service $X$ is optimal for an expiration time $\expiration =t$ if and only if 
it satisfies the following two conditions:
\begin{description}
\item{{\rm (a)}} $\omega(X_v,t)\ge \length(X_v)$ for each $v\in X$, and
\item{{\rm (b)}} $\omega(Z,t)\le \length(Z)$ for each subtree $Z$, disjoint with $X$,
			such that  $\parent(r_Z) \in X$.
\end{description}
\end{lemma}

\begin{proof}
$(\Rightarrow)$
We begin by proving that (a) and (b) are necessary conditions for optimality
of $X$.

(a) Suppose that there is a $v\in X$ for which $\omega(X_v,t) < \length(X_v)$.
Let $Y = X-X_v$. Then $Y$ is a service tree (empty if $v=r$), and we have
\begin{align*}
	\cost(Y,t) &= \length(Y) + \omega(\barY,t)
			\\
			&= \length(X) - \length(X_v) + \omega(\barX,t) + \omega(X_v,t)
			\\
			&< \length(X) + \omega(\barX,t) = \cost(X,t),
\end{align*}
contradicting the optimality of $X$.

(b) Suppose that there is a subtree $Z$ that violates condition (b), that is $Z\cap X = \emptyset$,
$\parent(r_Z) \in X$, but $\omega(Z,t) > \length(Z)$.
Let $Y= X\cup Z$. Then $Y$ is a service tree and
\begin{align*}
\cost(Y,t) &= \length(Y) + \omega(\barY,t)
 		\\
		&= \length(X) + \length(Z) + \omega(\barX,t) - \omega(Z,t)
		\\
		&< \length(X) + \omega(\barX,t) = \cost(X,t),
\end{align*}
contradicting the optimality of $X$.

$(\Leftarrow)$
We now prove sufficiency of conditions (a) and (b). Suppose that $X$ satisfies
(a) and (b), and let $Y$ be any other service subtree of $\calT$.
From (b), for any node $z\in Y-X$ with $\parent(z) \in X\cap Y$
we have $\omega(Y_z,t)\le \length(Y_z)$. 
Since both $X$ and $Y$ are rooted at $r$, any node in $Y-X$
is in some induced subtree $Y_z$, for some $z$ such that $\parent(z) \in X\cap Y$.
This implies that $\omega(Y-X,t)\le \length(Y-X)$.
Similarly, from (a), for any node $v\in X-Y$ with $\parent(v) \in X\cap Y$
we have $\omega(X_v,t)\ge \length(X_v)$. This implies that
$\omega(X-Y,t) \geq \length(X-Y)$. These inequalities  give us that
\begin{align*}
\cost(Y,t) &= \length(Y) + \omega(\barY,t)
			\\
		&= \length(X) + \omega(\barX,t)
				+ [\,\length(Y-X) - \omega(Y-X,t)\,]
				- [\,\length(X-Y) - \omega(X-Y,t)\,]
			\\
		&\ge \cost(X,t),
\end{align*}
proving the optimality of $X$.
\end{proof}

Following the terminology from Section~\ref{subsec: mlap notations}, 
a subtree $Z$ of $\calT$ (not necessarily rooted at $r$)
is called \emph{mature at time $t$} if $\omega(Z,t)\ge \length(Z)$. (We do not need to specify
the set of requests in $\omega(Z,t)$, as all requests are released at time $0$.)
In this section we will simplify this notation and write ``$t$-mature'', instead of
``mature at time $t$''.
We say that $Z$ is \emph{$t$-covered} if each induced subtree 
$Z_x$, for $x\neq r_Z$, is $t$-mature. (Note that in this definition
$Z$ itself is not required to be $t$-mature.)
We now make two observations. First, if $Z$ is $t$-covered then
the definition implies that each induced subtree $Z_v$ of $Z$ is $t$-covered as well.
Two, if $Z = \braced{r_Z}$, that is if $Z$ consists of only one node,
then $Z$ is vacuously $t$-covered; thus any subtree $Z$
of $\calT$ has a $t$-covered subtree rooted at $r_Z$. 


\begin{lemma}\label{lem: Ot unique}
If $X$ and $Y$ are $t$-covered service subtrees of $\calT$ then
the service subtree $X\cup Y$ is also $t$-covered.
\end{lemma}

\begin{proof}
If $X=Y$ the lemma is trivial, so assume $X\neq Y$.
Choose any $z\in (X-Y)\cup (Y-X)$ with $\parent(z) \in X\cap Y$.
Without loss of generality, we can assume that $z\in X-Y$.
By definition, $X_z$ is $t$-mature and disjoint with $Y$. 

Take $Q = Y \cup X_z$. $Q$ is a service subtree of $\calT$. 
We claim that $Q$ is $t$-covered. To justify this claim, choose any $v\in Q$. 
If $v\in Y$ and $z\notin Q_v$, then $Q_v$ is $t$-mature because $Q_v = Y_v$.
If $v\in Q_z = X_z$ then $Q_v$ is $t$-mature because $Q_v = X_v$.
The remaining case is when $v\in Y$ and $z\in Q_v$.
Then
$\omega(Q_v,t) = \omega(Y_v)+\omega(X_z,t) \ge \length(Y_v) + \length(X_z) = \length(Q_v,t)$,
so $Q_v$ is $t$-mature in this case as well. 
Thus indeed $Q$ is $t$-covered, as claimed.

We can now update $Y$ by setting $Y = Q$ and applying the above argument again.
By repeating this process, we will end up with $X = Y$, completing the proof.
\end{proof}

Choose $O^t$ to be the inclusion-maximal $t$-covered service subtree of $\calT$ (that is, 
a subtree rooted at $r$). By Lemma~\ref{lem: Ot unique}, $O^t$ is well defined and unique.
Also, from Lemma~\ref{lem: 1-phase optimal schedules} we obtain that $O^t$ is optimal for 
expiration time $\expiration = t$. Thus the optimal cost when $\expiration = t$ is
\begin{equation*}
	\opt(t) = \cost(O^t,t) = \length(O^t) + \omega(\barO^t,t).
\end{equation*}
Trivially, if a subtree $Z$ is $t$-mature and $t\le t'$ then $Z$ is $t'$-mature as well.
This implies the following corollary.


\begin{corollary}\label{cor:optima-grow}
For every $t \leq t'$ it holds that $O^t \subseteq O^{t'}$.
\end{corollary}




\subsection{An Online Competitive Algorithm}  

Without loss of generality, we can assume that
$\min_{v\in\calT-\braced{r}} \length_v > 1$; otherwise the distances
together with the waiting costs can be rescaled to satisfy this
property. To simplify the presentation we will assume that for
$\expiration\to\infty$ the optimum cost grows to $\infty$.  (Any
instance can be modified to have this property, without changing the
behavior of the algorithm on $\calT$, by adding an infinite path to
the root of $\calT$, where the nodes on this path have waiting cost
functions that are initially $0$ and then gradually increase.)


\paragraph{Algorithm~$\algDoubling$.}
For any $i\ge 0$, define $t_i$ to be the first time when $\opt(t_i) = 2^i$.
At each time $t_i$ serve $O^{t_{i+1}}$.

\smallskip

Algorithm~$\algDoubling$ is in essence a doubling algorithm~\cite{doubling-sigact}. However, although 
obtaining \emph{some} constant ratio using doubling is not difficult, the formulation that achieves 
the optimal factor of $4$ relies critically on the structure of optimal solutions that we elucidated 
earlier in this section. For example, note that the sequence of service costs of the algorithm 
does not necessarily grow exponentially.


\paragraph{Analysis.}

By our assumption that $\min_{v\in\calT-\braced{r}} \length_v > 1$, we have
$O^{t_0} = \braced{r}$; that is, until time $t_0$ the optimum solution will not make any services and will
only pay the waiting cost. This also implies that $\omega(O^{t_0},t_0)\le 1$.

We now estimate the cost of Algorithm~$\algDoubling$, for a given expiration time $\expiration$.
Suppose first that $\expiration = t_k$, by which we mean that the expiration is right after the algorithm's service at 
time $t_k$. The total service cost of the algorithm is trivially $\sum_{i=0}^{k} \length(O^{t_{i+1}})$.
To estimate the waiting cost, consider some node $v$. If $v\in O^{t_{i+1}}-O^{t_i}$, for some $i = 0,...,k$, then
the waiting cost of $v$ is $\omega(v,t_i)$. Otherwise, for $v\notin O^{t_{k+1}}$, the waiting cost of $v$
is $\omega(v,\theta) = \omega(v,t_k)$.
Thus {\algDoubling}'s total cost is
	\begin{align*}
		\alg(t_k) &= 		\sum_{i=0}^{k} \length(O^{t_{i+1}})
		+  \omega(O^{t_0},t_0) + \sum_{i=0}^{k} \omega(O^{t_{i+1}}-O^{t_i} , t_i) + \omega(\barO^{t_{k+1}},t_k) 
		\\
		&\le	\sum_{i=0}^{k} \big[\, \length(O^{t_{i+1}})	+  \omega(\barO^{t_i},t_i) \,\big] + 1
		\\
		&\le	\sum_{i=0}^{k+1} \big[\, \length(O^{t_{i}})	+  \omega(\barO^{t_i},t_i) \,\big] + 1
		\\
		&=  \sum_{i=0}^{k+1} \opt(t_i) + 1 = \sum_{i=0}^{k+1} 2^i + 1
						= 2^{k+2} = 4\cdot\opt(t_k),
	\end{align*}
as needed. 

Next, suppose that $\expiration$ is between two service times, say $t_k\le \expiration \le t_{k+1}$.
From the optimality of $O^{t_k}$ at expiration time $t_k$, we have
$\opt(t_k)= \length(O^{t_k}) + \omega(\barO^{t_k},t_k) 
			\le \length(O^\expiration) + \omega(\barO^\expiration,t_k)$.
Using this bound,
the increase of the optimum cost from time $t_k$ to time $\expiration$ can be estimated 
as follows:
\begin{align*}
	\opt(\expiration) - \opt(t_k) &\ge	
							\big[\, \length(O^\expiration) + \omega(\barO^\expiration,\expiration) \,\big]
		   					- \big[\, \length(O^\expiration) + \omega(\barO^\expiration,t_k) \,\big]
		\\
		&=  \omega(\barO^\expiration,\expiration) - \omega(\barO^\expiration,t_k)
		\ge  \omega(\barO^{t_{k+1}},\expiration) - \omega(\barO^{t_{k+1}},t_k),
\end{align*}
where the last expression is the increase in Algorithm~$\algDoubling$'s cost from time
$t_k$ to time $\expiration$. This implies that the ratio at expiration time $\expiration$ cannot be larger than the
ratio at expiration time $t_k$.

Finally, we have the case when $0 \le \expiration < t_0$. Thus $\opt(\expiration) < 1$. By our assumption that all
weights are greater than $1$, this implies that $\opt(\expiration) = \omega(\calT,\expiration)$, and thus
$\opt(\expiration)$ is the same as the cost of the algorithm.


Summarizing, we obtain our main result of this section.

\begin{theorem}\label{thm:doubling-for-phase}
	{\algDoubling} is $4$-competitive for the Single-Phase $\MLAP$.
\end{theorem}



\subsection{An Offline Polynomial-Time Algorithm}
\label{sec: 1p-mlap with deadlines}

The offline algorithm for computing the optimal solutions is based on
the above-established properties of optimal sets $O^t$.
It proceeds bottom up,
starting at the leaves, and pruning out subtrees that are not
$t$-covered.  The pseudo-code of our algorithm is shown below.


\begin{algorithm}[ht]
\caption{{\algCovSubT$(v,t)$}}\label{alg:paidsubtree}
\begin{algorithmic}[0]
\State $A_v \gets \braced{v}$
\State $\surplus_v \gets \omega(v,t)$
\For{each child $u$ of $v$}
    \State $\left(A_u,\surplus_u\right) \gets \algCovSubT(u,t)$
	\If{$\surplus_u\ge \length_u$}
    	\State $A_v \gets A_v \cup A_u$
    	\State $\surplus_v \gets \surplus_v + \surplus_u - \length_u$
	\EndIf
\EndFor
\State \textbf{return} $\left(A_v,\surplus_v\right)$
\end{algorithmic}
\end{algorithm}

For each node $v$ the algorithm outputs a pair
$(A_v,\surplus_v)$, where $A_v$ denotes the maximal (equivalently w.r.t. inclusion or cardinality)
$t$-covered subtree of $v$ rooted at $v$, and 
$\surplus_v = \omega(A_v,t) - \length(A_v-\braced{v})$, that
is $\surplus_v$ is the ``surplus'' waiting cost of $A_v$ at time $t$. (Note that
we do not account for $\length_v$ in this formula.)
To compute $O^t$, the algorithm returns {\algCovSubT$(r,t)$}.

By a routine argument, the running time of Algorithm~{\algCovSubT} is
$O(N)$, where $N$ is the size of the instance (that is, the number of nodes
in $\calT$ plus the number of requests). Here, we assume that the
values $\omega(v,t)$ can be computed in time proportional to the number of
requests in $v$.


\section{{\MLAP} on Paths}
\label{sec: mlap on paths}

We now consider the case when the tree is just a path. For simplicity we will
assume a~generalization to the~continuous case, that we refer to as 
\emph{the {\MLAP} problem on the line}, when the path is represented by the half-line
$[0,\infty)$; that is the requests can occur at any point $x\in [0,\infty)$. Then
the point $0$ corresponds to the root, each node is a point $x\in [0,\infty)$,
and each service is an interval of the form $[0,x]$.
We say that an algorithm {\em delivers from}~$x$ if it serves the interval $[0,x]$.

We provide several results for the {\MLAP} problem on the line.
We first prove
that the competitive ratio of {\MLAPD} (the variant with deadlines) on
the line is exactly $4$, by providing matching upper and lower
bounds. Then later we will show that the lower bound of $4$ can be
modified to work for {\MLAPL} (that is, for linear waiting costs).


\paragraph{Algorithm~{\DLINE}.}

The algorithm creates a service only when a deadline of a pending request is reached.
If a~deadline of a request at $x$ is reached, then {\DLINE} delivers from $2x$.


\begin{theorem}
Algorithm~{\DLINE} is $4$-competitive for {\MLAPD} on the line.
\end{theorem}

\begin{proof}
The proof uses a charging strategy. We represent each adversary service, say when the 
adversary delivers from a point $y$, by an interval $[0,y]$. The cost of each
service of {\DLINE} is then charged to a segment of one of those adversary service intervals.

Consider a service triggered by a deadline $t$ of a request $\rho$ at some point $x$.
When serving $\rho$, {\DLINE} delivered from $2x$.
The adversary must have served $\rho$ between its arrival time and its
deadline~$t$. Fix the last such service of the adversary, where at a time $t'\le t$ the
adversary delivered from a point $x'\ge x$. We charge the
cost $2x$ of the algorithm's service to the segment $[x/2,x]$ of the
adversary's service interval $[0,x']$ at time $t'$. 

We now claim that no part of the adversary's service is charged twice.
To justify this claim, suppose that there are two services of {\DLINE}, at 
times $t_1 < t_2$, triggered by requests from points $x_1$ and $x_2$, respectively, 
that both charge to an adversary's service from $x'$ at time $t' \leq t_1$. 
By the definition of charging, the request at $x_2$ was 
already present at time $t'$. As $x_2$ was not served 
by {\DLINE}'s service at $t_1$, it means that $x_2 > 2 x_1$, and thus
the charged segments $[x_1/2,x_1]$ and $[x_2/2,x_2]$ of the adversary service
interval at time $t'$ are disjoint.

Summarizing, for any adversary service interval $[0,y]$, its charged
segments are disjoint. Any charged segment receives the charge equal to
$4$ times its length. Thus this interval receives the total charge
at most $4y$. This implies that the competitive ratio is at most~$4$.
\end{proof}


\paragraph{Lower bounds.}

We now show lower bounds of $4$ for {\MLAPD} and {\MLAPL} on the line.
In both proofs we show the bound for the corresponding
variant of {\SPMLAP}, using a reduction from the online bidding
problem~\cite{doubling-sigact,online-bidding}. Roughly speaking, in online
bidding, for a given universe $\mathcal{U}$ of real numbers, the adversary chooses a secret
value $u \in \mathcal{U}$ and the goal of the algorithm is to find an upper-bound
on $u$. To this end, the algorithm outputs an increasing sequence of numbers
$x_1, x_2, x_3, \ldots$. The game is stopped after the first $x_k$ that is at
least $u$ and the bidding ratio is then defined as $\sum_{i=1}^k x_i / u$.

Chrobak et al.~\cite{online-bidding} proved that the optimal bidding ratio is exactly $4$, even if it is 
restricted to sets $\mathcal{U}$  of the form $\{1,2,\ldots,B\}$, for some integer $B$. 
More precisely, they proved the following result.

\begin{lemma}
\label{lem:online_bidding}
For any $R < 4$, there exists $B > 0$, such that any sequence of integers 
$0 = x_0 < x_1 < x_2 < \ldots < x_{m-1} < x_m = B$ has an index $k \geq 1$ with
$\sum_{i=0}^k x_i > R \cdot (x_{k-1} + 1)$.
\end{lemma}

\begin{theorem}
\label{thm:mlapd_path_lower_bound}
There is no online algorithm for {\MLAPD} on the line with competitive ratio smaller than $4$.
\end{theorem}

\begin{proof}
We show that no online algorithm for {\SPMLAPD} (the deadline variant of {\SPMLAP}) 
on the line can attain competitive ratio smaller than $4$. 
Assume the contrary, i.e., that there exists a deterministic algorithm $\ALG$ that is $R$-competitive,
where $R < 4$. 
Let $B$ be the integer whose existence is guaranteed by Lemma~\ref{lem:online_bidding}.
We create an instance of {\SPMLAPD}, where, at time $0$,
for every $x \in \{1,\ldots,B\}$  there is a~request at $x$ with deadline $x$.

Without loss of generality, {\ALG} issues services only at integer times $1,2,...,B$.
The strategy of {\ALG} can be now defined as a sequence of services at times 
$t_1 < t_2 < \ldots < t_m$, where at time $t_i$ it delivers from 
$x_i \in \braced{t_i,t_i+1,...,B}$.  Without loss of
generality, $x_1 < x_2 < \ldots < x_m$. We may assume that $x_m = B$ (otherwise
the algorithm is not competitive at all); we also add a~dummy service from $x_0 = 0$ at time $t_0 = 0$.

The adversary now chooses some $k \geq 1$ and stops the game at the expiration time that is right after
the algorithm's $k$th service, say $\expiration = t_k+\half$.
{\ALG}'s cost is then $\sum_{i=0}^k x_i$. 
The request at $x_{k-1}+1$ is not served at time $t_{k-1}$, so, to meet the deadline of this request,
the schedule of $\ALG$ must satisfy $t_k \leq x_{k-1}+1$.  
This implies that $\expiration < x_{k-1}+2$, that is,
all requests at points $x_{k-1}+2, x_{k-1}+3,...,B$ expire 
before their deadlines and do not need to be served.
Therefore, to serve this instance, 
the optimal solution may simply deliver from $x_{k-1}+1$ at time~$0$. 
Hence, the competitive ratio of {\ALG} is at least 
$\sum_{i=0}^k x_i / (x_{k-1}+1)$. By Lemma~\ref{lem:online_bidding}, it is possible to choose~$k$
such that this ratio is strictly greater than $R$, a contradiction with 
$R$-competitiveness of {\ALG}.
\end{proof}

Next, we show that the same lower bound applies to {\MLAPL}, the version of
{\MLAP} where the waiting cost function is linear. This improves the lower
bound of $3.618$ from~\cite{aggregation_wads_2013}.


\begin{theorem}
\label{thm:mlapl_path_lower_bound}
There is no online algorithm for {\MLAPL} on the line with competitive ratio smaller than $4$.
\end{theorem}

\begin{proof}
Similarly to the proof of Theorem~\ref{thm:mlapd_path_lower_bound}, we create an instance of 
{\SPMLAPL} (the variant of {\SPMLAP} with linear waiting cost functions)
that does not allow a better than $4$-competitive online algorithm.
Fix any online algorithm {\ALG} for $\SPMLAPL$
and, towards a contradiction, suppose that it is $R$-competitive, for some $R < 4$.
Again, let $B$ be the integer whose existence is guaranteed by Lemma~\ref{lem:online_bidding}.
In our instance of {\SPMLAPL}, there are $6^{B-x}$ requests at $x$ for any 
$x \in \braced{1, 2,\ldots, B}$. 

Without loss of generality, we make the same assumptions as in the proof of
Theorem~\ref{thm:mlapd_path_lower_bound}: algorithm {\ALG} is defined by
a sequence of services at times $0 = t_0 < t_1 < t_2 < \ldots < t_m$,
where at each time $t_i$ it delivers from some point $x_i$. Without loss of generality,
we can assume that $0 = x_0 < x_1 < \ldots < x_m = B$. 

Again, the strategy of the adversary is to stop the game at some expiration time $\expiration$
that is right after some time $t_k$, say $\expiration = t_k + \epsilon$, for some small $\epsilon > 0$.
The algorithm pays $\sum_{i=0}^k x_i$ for serving the requests. The  
requests at $x_{k-1} + 1$ waited for time~$t_k$ in {\ALG}'s schedule 
and hence {\ALG}'s waiting cost is at least $6^{B-x_{k-1}-1} \cdot t_k$. 

The adversary delivers from point $x_{k-1} + 1$ at time $0$.
The remaining, unserved requests at points $x_{k-1} + 2, x_{k-1} + 3, \ldots, B$
pay time $\expiration$ each for waiting. 
There are $\sum_{j=x_{k-1}+2}^B 6^{B-j} 
			\leq  \onefifth\cdot 6^{B-x_{k-1}-1}$ 
such requests and hence the adversary's waiting cost is at most 
$\onefifth\cdot 6^{B-x_{k-1}-1} \cdot (t_k+\epsilon)$.

Therefore, the algorithm-to-adversary ratio on the waiting costs is 
at least $5t_k/(t_k+\epsilon)$. For any $k$ we can choose a sufficiently small $\epsilon$
so that this ratio is larger than $4$.
By Lemma~\ref{lem:online_bidding}, it is possible to choose~$k$ for which
the ratio on servicing cost is strictly greater than $R$. 
This yields a contradiction to the $R$-competitiveness of {\ALG}.
\end{proof}

We point out that
the analysis in the proof above gives some insight into the behavior of any
$4$-competitive algorithm for {\SPMLAPL} (we know such an algorithm
exists, by the results in Section~\ref{sec: one-phase MLAP}), namely that,
for the type of instances used in the above proof,
its waiting cost must be negligible compared to the service cost.


\section{An Offline 2-Approximation Algorithm for {\MLAPD}}
\label{sec: mlap with deadlines}

In this section we consider the offline version of $\MLAPD$, for which Becchetti~\etal~\cite{packet-aggregation-becchetti}
gave a polynomial-time $2$-approximation algorithm based on LP-rounding. We give a
much simpler argument that does not rely on linear programming. 

We will use an alternative specification of schedules that
is easier to reason about in the context of offline approximations.
If $\schedS$ is a schedule, for each node
$x\in \calT$ we can specify the set $\schedS_x$ of times $t$ for which
$\schedS$ contains a service $(X,t)$ with $x\in X$. Then the set
$\braced{\schedS_x}_{x\in \calT}$ uniquely determines $\schedS$.  Note
that we have $\schedS_x\subseteq \schedS_y$ whenever $y$ is the parent
of $x$.  Further, we can now write the service cost as
$\scost(\schedS) = \sum_{x\in\calT} |\schedS_x|\length_x$.  It is easy to
see that (without loss of generality) in an optimal (offline) schedule
$\schedS$ each service time is equal to some deadline, and we will
make this assumption in this section; in particular,
$\schedS_{r}$ can be assumed to be the set of all deadlines.

Let $\calJ$ be the given instance. For each node $v$, define $\calJ_v$ to be the set of all intervals $[a_\rho,d_\rho]$, 
for requests $\rho$ issued in $\calT_v$. 


\paragraph{Algorithm~$\algLBL$.}

We proceed level by level, starting at the root and in order of increasing depth, computing the service 
times $\schedS_v$ for all nodes $v \in \calT$. 
For the root $r$, $\schedS_r$ is the set of the deadlines of all requests.
Consider now some node $v$ with parent $u$ for which $\schedS_u$ has already been computed. 
Using the standard earliest-deadline algorithm, compute $\schedS_v$ as the minimum cardinality subset of $S_u$
that intersects all intervals in $\calJ_v$.

\medskip


Algorithm~$\algLBL$ clearly runs in polynomial time; in fact it can be implemented in time 
$O(N\log N)$, where $N$ is the total size of $\calJ$.

We now show that the approximation ratio of Algorithm~$\algLBL$ is at most $2$. (It is easy to find an
example showing that this ratio is not better than $2$.)
Denote by $\optschedS$ an optimal schedule for $\calJ$. According to our
convention, $\optschedS_v$ is then the set of times when $v$ is served in $\optschedS$. Since
$\cost(\schedS) = \sum_v \length_v|\schedS_v|$ and the optimum cost 
is $\cost(\optschedS) = \sum_v \length_v|\optschedS_v|$, it
is sufficient to show that $|\schedS_v|\le 2|\optschedS_v|$ for each $v\neq r$.
This is quite simple: if $u$ is the father of $v$ then $\schedS_u$ intersects all intervals in $\calJ_v$. 
We construct $\schedS'_v\subseteq \schedS_u$ as follows.
For each $t\in \optschedS_v$, choose the maximal $t^-\in \schedS_u$ such that  $t^-\le t$, and
the minimal $t^+\in \schedS_u$ such that  $t^+\ge t$.
Add $t^-, t^+$ to $\schedS'_v$. (More precisely, each of them is added only if it is defined.)
Then  $\schedS'_v\subseteq \schedS_u$ and $|\schedS'_v| \le 2|\optschedS_v|$.
Further, any interval $[a_\rho,d_\rho]\in \calJ_v$ contains some $t\in \optschedS_v$ and intersects $\schedS_u$,
so it also must contain either $t^-$ or $t^+$. Therefore  $\schedS'_v$ intersects all intervals in $\calJ_v$. Since
we pick $\schedS_v$ optimally from $\schedS_u$, we have $|\schedS_v|\le |\schedS'_v| \le  2|\optschedS_v|$,
completing the proof.


\section{General Waiting Costs}
\label{sec: general waiting costs}

Our model of {\MLAP} assumes full continuity, namely that the time is
continuous and that the waiting costs are continuous functions of
time, while in some earlier literature authors use the discrete model.
Thus we still need to show that our algorithms can be applied in the
discrete model without increasing their competitive ratios. We also
consider the model where some request may remain unserved.  We
explain how our results can be extended to these models as well. 
We will also show that our results can be extended to functions that
are left-continuous, and that {\MLAPD} can be represented as a special
case of {\MLAP} with left-continuous functions. While those reductions
seem intuitive, they do involve some pesky technical challenges, and they
have not been yet formally treated in the literature.


\paragraph{Extension to the discrete model.} 

In the discrete model (see~\cite{jrp-online-buchbinder}, for example), 
requests arrive and services may happen only at integral
points $t=1,\ldots,H$, where $H$ is the time horizon. The waiting
cost functions $\omega_\rho$ are also specified only at integral
points. (The model in~\cite{jrp-online-buchbinder} also allows waiting costs to be
non-zero at the release time. However we can assume that $\omega_\rho(a_\rho)=0$, since
increasing the waiting cost function uniformly by an additive constant can only decrease the
competitive ratio.)

We now show how to simulate the discrete time model in the model where time
and waiting costs are continuous. Suppose that $\calA$ is an
$R$-competitive online algorithm for the model with continuous time and continuous 
waiting cost functions.  We construct an $R$-competitive
algorithm $\calB$ for the discrete time model.

Let $\calJ=\angled{\calT,\calR}$ be an instance given to $\calB$. We
extend each waiting cost function $\omega_\rho$ to non-integral
times as follows: for each integral $t=a_\rho,\ldots,H-1$ we define
$\omega_\rho(\tau)$ for $\tau\in(t,t+1)$ so that it continuously
increases from $\omega_\rho(t)$ to $\omega_\rho(t+1)$ (e.g., by linear interpolation);
$\omega_\rho(\tau)=0$ for all $\tau<a_\rho$; and
$\omega_\rho(\tau)=\omega_\rho(H)$ for all $\tau>H$. 

Algorithm $\calB$ presents the instance $\calJ=\angled{\calT,\calR}$
with these continuous waiting cost functions to $\calA$. At each
integral time $t=1,\ldots,H-1$, $\calB$ simulates $\calA$ on the whole
interval $[t,t+1)$.  If $\calA$ makes one or more services, $\calB$
makes a single service at time $t$ which is their union. This is
possible, since no request arrives in $(t,t+1)$. At time $H$,
algorithm $\calB$ issues the same service as $\calA$.

Overall, $\calB$ produces a feasible schedule in the discrete time model. 
The cost of $\calB$ does not exceed the cost of
$\calA$.  On the other hand, any feasible (offline) schedule $\schedS$
in the discrete time model is also a feasible schedule in the
continuous time model with the same cost. Thus $\calB$ is
$R$-competitive.


\paragraph{Unserved requests with bounded waiting costs.}

In our definition of {\MLAP} we require that all the requests are
eventually served. However, if the waiting cost of a request $\rho$ is
bounded, it is natural to allow a possibility that $\rho$ is not
served in  a schedule $\schedS$; in that case it incurs waiting cost 
$\wcost(\rho,\schedS)=\lim_{t\rightarrow+\infty}\omega_\rho(t)$. In
this variant, there is no time horizon in the instance.

Our algorithm {\OnAlgTreesGeneral} works in this model as well, with
the competitive ratio increased at most by one. The only modification of
the algorithm is that there is no final service at the time
horizon.  Instead we let the time proceed to infinity, issuing
services at the maturity times of $q$ (the quasi-root of $\calT$). 

To modify our charging scheme to this variant, the key observation is
that if a node $v$ is never serviced both in {\OnAlgTreesGeneral}
and in an optimal schedule $\optschedS$, then the requests at $v$ pay
the same waiting costs in both schedules. Thus we can ignore such
nodes and requests at them. We claim that for each remaining node
$v$, the pseudo-schedule $\pseudoschedS$ contains at least one
pseudo-service of $v$: Indeed, otherwise $v$ is not served in
$\optschedS$ and the total (limit of the) waiting cost of all the
(unserved) requests in the induced subtree $\calT_v$ is less than $\length_v$,
which implies that the maturity time of $v$ is always infinite and
thus $v$ is never serviced in {\OnAlgTreesGeneral} either,
contradicting the fact that $v$ was not ignored before. Now consider
all the remaining unserved requests and add to the schedule of
{\OnAlgTreesGeneral} one last service that serves all these requests.
As the unserved requests do not cause $q$ to mature, this increases
the cost of {\OnAlgTreesGeneral}; at the same time the service of
each node can be charged to a pseudo-service of the same node
in $\pseudoschedS$, which increases the competitive ratio by at most~1.


\paragraph{Extension to left-continuous waiting costs.}  

We now argue that we can modify our algorithms to handle
\emph{left-continuous} waiting cost functions, i.e., functions that
satisfy $\lim_{\tau \nearrow t} \omega_\rho(\tau) = \omega_\rho(t)$
for each time $t\ge 0$.  Left-continuity enables an online algorithm
to serve a request at the last time when its waiting cost is at or
below some given threshold.

Some form of left-continuity is also necessary for constant
competitiveness.  To see this, think of a simple example of a tree of
depth $1$ and with $\length_q = 1$, and a sequence of requests in $q$ with release times
approaching $1$, and waiting cost functions defined by
$\omega_\rho(1)=K\gg 1$ and $\omega_\rho(t)=0$ for $t<1$. If
an online algorithm serves one such request before time $1$, the
adversary immediately releases another. The sequence stops either
after $K$ requests or after the algorithm serves some
request at or after time $1$, whichever comes first. 
The optimal cost is at most $\length_q=1$,
while the online algorithm pays at least $K$.

The basic (but not quite correct) idea of our argument for left-continuous waiting
cost functions is this:
For any time point $h$ where some waiting cost function has a discontinuity,
we replace point $h$ by a ``gap interval'' $[h,h+\epsilon]$, for some $\epsilon > 0$. 
The release times after time $h$ and the values of all
waiting cost functions after $h$ are shifted to the right by $\epsilon$.
In the interval $[h,h+\epsilon]$, for each request $\rho$, its waiting
cost function is filled in by any non-decreasing continuous curve with value
$\nu^-$ at $h$ and $\nu^+$ at $h+\epsilon$, for $\nu^- = \omega_\rho(h)$
and $\nu^+ = \lim_{\tau\searrow h} \omega_\rho(\tau)$.
Thus the waiting cost functions that are continuous at $h$ are simply
``stretched'' in this gap interval, where their values remain constant.
This will convert the original instance $\calJ$ into an instance
$\calJ'$ with continuous waiting cost functions; then we can apply
a simulation similar to the one for the discrete model, with the
behavior of an algorithm $\calA$ on $\calJ'$ inside $[h,h+\epsilon]$ mimicked
by the algorithm $\calB$ on $\calJ$ while staying at time $h$.

The above construction, however, has a flaw: as $\calB$ is online, for each
newly arrived request $\rho$ it would need to know the future requests
in order to correctly modify $\rho$'s waiting cost function (which needs
to be fully revealed at the arrival time). Thus, inevitably, $\calB$ will
need to be able to modify waiting cost functions of earlier
requests, but the current state of $\calA$ may depend on these functions.
Such changes could make the computation of $\calA$ meaningless. 
To avoid this problem, we
will focus only on algorithms $\calA$ for continuous cost functions that
we call \emph{stretch-invariant}. Roughly, those are algorithms whose
computation is not affected by the stretching operation described above.

To formalize this, let $\bbI = \braced{[h_i,h_i+\eps_i] \mid i=1,\ldots, k}$
be a finite set of \emph{gap intervals}, where all times $h_i$ are distinct.
(For now we can allow the $\eps_i$'s to be any positive reals; their purpose will
be explained later.)
Let $\shift(t,\bbI)=t+\sum_{i:h_i<t}\eps_i$ denote the time $t$ shifted right by
inserting intervals $\bbI$ on the time axis. We extend this operation to
requests in a natural way: for any request $\rho$ with a 
continuous waiting cost function,
$\shift(\rho,\bbI)$ denotes the request modified by inserting $\bbI$
on the time axis and filling in the values of $\omega_\rho$ in the
inserted intervals by constant functions, as described earlier. For a set
of requests $P\subseteq \calR$,  the stretched set of requests
$\shift(P,\bbI)$ is the set consisting of requests $\shift(\rho,\bbI),$ for
all $\rho\in P$.

Consider an online algorithm $\calA$ for {\MLAP} with
continuous waiting cost functions.
We say that $\calA$ is {\em stretch-invariant} if for every instance $\calJ =
\angled{\calT,\calR}$ and any set of gap intervals $\bbI$, the
schedule produced by $\calA$ for the instance
$\angled{\calT,\shift(\calR,\bbI)}$ is obtained from the schedule
produced by $\calA$ for $\calJ$ by shifting it according to $\bbI$,
namely every service $(X,t)$ is
replaced by service $(X,\shift(t,\bbI))$.

Most natural algorithms for {\MLAP} are stretch-invariant. In case of
{\OnAlgTreesGeneral}, observe that its behavior depends only on the
maturity times $M_P(v)$ where $P$ is the set of pending requests and
$M_{\shift(P,\bbI)}(v)=\shift(M_P(v),\bbI)$; in particular
stretching does not change the order of the maturity times. 
Using induction on the current time $t$, we observe 
{\OnAlgTreesGeneral} creates a service $(X,t)$ in its schedule for the request set $\calR$
if and only if
{\OnAlgTreesGeneral} creates a service $(X,\shift(t,\bbI))$ in its schedule
for the request set $\shift(\calR,\bbI)$.

Suppose that $\calA$ is an $R$-competitive online
algorithm for continuous waiting cost functions that is stretch-invariant.  
We convert $\calA$ into an
$R$-competitive algorithm $\calB$ for left-continuous waiting costs.
Let $\calJ=\angled{\calT,\calR}$ be an instance given to $\calB$.
Algorithm $\calB$ maintains the set of gap intervals $\bbI$, and
a set of requests $\calP$ presented to $\calA$; both sets are
initially empty.  Algorithm $\calB$ at time $t$ simulates the
computation of $\calA$ at time $\shift(t,\bbI)$.

If a new request $\rho\in\calR$ is released at time $t=a_\rho$,
algorithm $\calB$ obtains $\rho'$ from $\shift(\rho,\bbI)$ by
replacing the discontinuities of $\omega_\rho$ by new gap intervals $\bbI_\rho$
on which $\omega_{\rho'}$ is defined so that it continually
increases. (If a gap interval already exists in $\bbI$ at the
given point, it is used instead of creating a new one, to maintain the
starting points distinct.) We set
$a_{\rho'}=\shift(t,\bbI)$, which is the current time in $\calA$. We update
$\bbI$ to $\bbI\cup\bbI_\rho$; this does not change the current
time in $\calA$ as all  new gap intervals start at or after
$t$.  We stretch the set of requests $\calP$ by $\bbI_\rho$; this
does not change the past output of $\calA$, because $\calA$ is
stretch-invariant. (Note that the state of $\calA$ at time $t$ may change, but
this does not matter for the simulation.)
Finally, we add the new request $\rho'$ to $\calP$. 

If the current time $t$ in $\calB$ is at a start point of a
gap interval, i.e., $t=h_i$, algorithm $\calB$ simulates the
computation of $\calA$ on the whole shifted gap interval
$\angled{\shift(h_i,\bbI),\eps_i}$. If $\calA$ makes one or more
services in $\angled{\shift(h_i,\bbI),\eps_i}$, $\calB$ makes a
single service at time $t$ which is their union.

The cost of $\calB$ for requests $\calR$ does not exceed the cost of
$\calA$ for requests $\calP$.  Any adversary schedule $\schedS$ for
$\calR$ induces a schedule $\schedS'$ for $\calP$ with the same
cost. Since $\calA$'s cost is at most $R\cdot \cost(\schedS')$, we
obtain that $\calB$'s cost is at most $R\cdot\cost(\schedS)$;
hence $\calB$ is $R$-competitive.

In the discussion above we assumed that the instance has a finite number
of discontinuities. Arbitrary left-continuous waiting
cost functions may have infinitely many discontinuity points, but the
set of these points must be countable. The construction described
above extends to arbitrary left-continuous cost functions, as long as
we choose the $\eps_i$ values so that their sum is finite.


\paragraph{Reduction of {\MLAPD} to {\MLAP}.}

We now argue that $\MLAPD$ can be expressed as a
variant of $\MLAP$ with left-continuous waiting cost functions. The idea
is simple: a
request $\rho$ with deadline $d_\rho$ can be assigned a waiting cost
function $\omega_\rho(t)$ that is $0$ for times $t\in [0,d_\rho]$ and $\infty$
for $t > d_\rho$ -- 
except that we cannot really use $\infty$, so we need to replace it by
some sufficiently large number.
If $\trignode_\rho = v$, we let 
$\omega_\rho(t) = \length^\ast_v$, where $\length^\ast_v$ is the sum of all
weights on the path from $v$ to $r$ (the ``distance'' from $v$ to $r$).
This will convert an instance $\calJ$ of $\MLAPD$ into an
instance $\calJ'$ of $\MLAP$ with left-continuous waiting cost functions.

We claim that, without loss of generality, any online algorithm
$\calA$ for $\calJ'$ serves any request $\rho$ before or at time $d_\rho$.
Otherwise, $\calA$ would have to pay waiting cost of $\length^\ast_v$ for $\rho$
(where $v = \trignode_\rho$),
so we can modify $\calA$ to serve $\rho$ at time $d_\rho$ instead,
without increasing its cost.
We can then treat $\calA$ as an algorithm for $\calJ$.
$\calA$ will meet all deadlines in $\calJ$ and its cost on $\calJ$ will be
the same as its cost on $\calJ'$, which means that
its competitive ratio will also remain the same.

Note that algorithm {\OnAlgTreesGeneral} (or rather its extension to
the left-continuous waiting costs, as described above) does
not need this modification, as it already guarantees that when the
waiting cost of a request at $v$ reaches $\length^\ast_v$, all the nodes on
the path from $v$ to $r$ are mature and thus the whole path is served.

\bibliographystyle{abbrv}
\bibliography{references}

\end{document}